\documentclass[preprint,12pt,authoryear]{elsarticle}

\usepackage{siunitx}
\usepackage{amssymb}
\usepackage{natbib}
\usepackage{graphicx}
\usepackage{float}
\usepackage{amsfonts}
\usepackage{amsmath}
\usepackage{subcaption}
\usepackage{amsthm}    
\usepackage{xcolor}
\usepackage{tikz}
\usetikzlibrary{patterns,decorations.pathmorphing,positioning}
\newtheorem{theorem}{Theorem}

\theoremstyle{definition}

\newtheorem{corollary}{Corollary}
\newtheorem{lemma}{Lemma}

\theoremstyle{remark}
\newtheorem{remark}{Remark}

\theoremstyle{definition}

\theoremstyle{definition}
\newtheorem*{problem}{Problem}

\journal{}

\begin{document}

\begin{frontmatter}



\title{Model Reference Adaptive Control of Piecewise Affine Systems with State Tracking Performance Guarantees}


\author{Tong Liu, Martin Buss}

\address{Technische Universit{\"a}t M{\"u}nchen, Munich, 80333, Germany}

\begin{abstract}
In this paper, we investigate the model reference adaptive control approach for uncertain piecewise affine systems with performance guarantees. The proposed approach ensures the error metric, defined as the weighted Euclidean norm of the state tracking error, to be confined within a user-defined time-varying performance bound. We introduce an auxiliary performance function to construct a barrier Lyapunov function. This auxiliary performance signal is reset at each switching instant, which prevents the transgression of the barriers caused by the jumps of the error metric at switching instants. The dwell time constraints are derived based on the parameters of the user-defined performance bound and the auxiliary performance function. We also prove that the Lyapunov function is non-increasing even at the switching instants and thus does not impose extra dwell time constraints. Furthermore, we propose the robust modification of the adaptive controller for the uncertain piecewise affine systems subject to unmatched disturbances. A Numerical example validates the correctness of the proposed approach.
\end{abstract}




\end{frontmatter}


\section{Introduction}
\label{sec: intro}
The study of piecewise affine systems (PWA) systems has been of significant interest due to their capability to approximate nonlinear systems and model hybrid systems. A PWA system consists of several linear subsystems. Each subsystem is associated with a certain region in the state-input space. Depending on in which region the state-input vector lies, the PWA system is governed by the associated subsystem dynamics. When the state-input trajectory goes through the boundary of two neighbouring regions (described mathematically by hyperplanes), the switching from one subsystem to another subsystem is triggered. Early studies focus on the controllability and observability \cite{bemporad2000observability,collins2004observability}, convergence analysis \cite{pavlov2007convergence}, and control synthesis \cite{rodrigues2003observer,habets2006reachability}, where the system parameters and region partitions are exactly known.

In the physical world, an exact system model is mostly not accessible due to uncertainties and disturbances. Therefore, introducing the adaptive mechanism into the uncertain PWA systems has significant meaning, especially when the uncertainties and disturbances are so large that a single robust controller cannot stabilize the closed-loop system. 
Due to the hybrid nature of the PWA systems, not only the uncertain parameters need to be estimated by designing adaptation laws, but also the switching behavior of the closed-loop system needs to be carefully considered.
In the last decade, model reference adaptive control (MRAC) approaches have been investigated for uncertain PWA systems. 
The methods proposed in work of di Bernardo \textit{et al.} \cite{di2013hybrid, bernardo2013model, di2016extended} rely on common Lyapunov functions, where the closed-loop systems are allowed to switch arbitrarily fast. MRAC for piecewise linear (PWL) systems, a special version of the PWA systems, are investigated in work of Sang and Tao \cite{sang2011adaptive2, sang2012adaptive2}, where the dwell-time constraints for switches are given to ensure the closed-loop stability. Its extension to PWA systems is reported recently \cite{kersting2017direct}, where the exponential decaying of the state tracking error is proved given that a persistently exciting (PE) condition and some dwell time constraints are fulfilled. To enhance the robustness of the adaptive switched systems against disturbances and time-delay, some robust MRAC approaches have been proposed for switched linear systems, whose formulation is similar to PWA systems but with switching signals given externally. These include robust MRAC with dead zone \cite{wang2012model} and leakage \cite{yuan2018robust}, robust $H_{\infty}$ MRAC \cite{wu2015h,xie2018h} as well as MRAC with asynchronous switching between subsystems and controllers \cite{wu2015adaptive,yuan2018novel}. 

Despite the aforementioned advances, the adaptive control for PWA systems fulfilling a user-defined performance guarantee (such as state constraints) is rarely studied. In light of the fact that a lot of systems in practice have constraints like physical or operational boundaries, saturation, performance and safety specifications, we would like to explore the MRAC of PWA systems with state tracking performance guarantees.

Notable progress has been made in the field of performance guarantees with adaptive control methods. These include funnel control \cite{ilchmann2009pi, hackl2013funnel}, barrier Lyapunov function-based approach \cite{tee2009barrier} and prescribed performance control \cite{bechlioulis2008robust,bechlioulis2010prescribed}. All of these methods are proposed to confine the output tracking error within the predefined constraints. Although some recent barrier Lyapunov function-based controllers achieve the full state constraints \cite{liu2014adaptive, liu2016barrier, zhao2018removing, niu2015new}, they are built upon the backstepping structure, which requires the controlled system to be in strict feedback form or pure feedback form. This prevents their application to the generalized PWA systems.

Recently, a set-theoretic based MRAC for linear systems is developed\cite{arabi2018set}. It uses the barrier Lyapunov function concept to confine the weighted Euclidean norm of the state tracking error within a predefined bound. The controller does not rely on the backstepping-type analysis and therefore does not imposes restrictions on the system structure. This method is extended to the cases with time-varying performance bound \cite{arabi2019set}, system with actuator faults \cite{xiao2019robust} and systems with unstructured uncertainties \cite{arabi2019neuroadaptive}. However, applying this method to switched systems is nontrivial and challenging. If the barrier function is constructed with the user-defined performance bound being the barrier, as it is done in the linear system case, then the discontinuity of the weighted Euclidean norm of the tracking error at switching instants may cause transgression of the barrier, which makes the barrier function invalid. Besides, only matched uncertainties (uncertainties, which can be compensated with an additional input term) are addressed in the work of set-theoretic MRAC approaches. Since the PWA systems are mostly approximation of nonlinear systems, their approximation errors are not necessarily matched, let alone other kinds of external disturbances.

The main contribution of this paper is twofold. First, a set-theoretic MRAC approach for uncertain PWA systems with state tracking performance guarantees is developed. Second, a robust modification of this method is proposed for PWA systems subject to unmatched disturbances. Specifically, we impose an auxiliary performance signal with a state reset map to construct the barrier function, which bypasses the barrier transgression problem.
The dwell time constraints are derived based on the auxiliary performance signal and the user-defined performance bound. The Lyapunov function is non-increasing, even at switching instants and therefore, does not impose extra dwell time constraints. Furthermore, a projection-based robust modification of the proposed approach is developed to enhance the robustness against disturbances. Compared with the state-of-the-art set-theoretic MRAC approaches, the disturbances are not required to be matched and boarder application is achieved.

The paper is structured as follows. The definition of PWA systems, MRAC and the performance function are revisited in Section \ref{sec: pre}. The proposed method is explained in Section \ref{sec: method}, in which the stability analysis is also provided. A numerical example is illustrated in Section \ref{sec: validation}. 

\textit{Notations:} In this paper, $\mathbb{R}, \mathbb{R^+}$ and $\mathbb{N}^+$ denote the set of real numbers, positive real numbers and positive natural numbers, respectively. $\mathrm{tr}(\cdot)$ represents the trace of a matrix. The Euclidean norm is denoted by $\|\cdot\|_2$. $\lambda_{\mathrm{max}}(P)$ and $\lambda_{\mathrm{min}}(P)$ represent the maximal and minimal eigenvalues of matrix $P$, respectively.

\section{Preliminaries and Problem Statement}
\label{sec: pre}
Consider the nonlinear system
\begin{equation}
    \dot{x}(t) = g(x(t), u(t)),
\end{equation}
where $x \in \mathbb{R}^n$ and $u \in \mathbb{R}^p$ denote its state and control input signal. $g: \mathbb{R}^{n+p} \to \mathbb{R}^n$ represents a smooth nonlinear function. 
Given a set of operating points $(x_i^*, u_i^*), i \in \mathcal{I} \triangleq\{1,2,\cdots,s\}$, the state-input space $[x^T, u^T]^T \in \mathbb{R}^{n+p}$ can be divided into $s$ convex regions $\{\Omega_i\}_{i=1}^s$. Each operating point locates at the center of each region. For every time instant $t$, the state-input vector $[x^T(t),u^T(t)]^T$ can only belong to one region. The regions have no overlaps, i.e., $\Omega_i \cap \Omega_j = \emptyset$ for $i \neq j$ and $i,j \in \mathcal{I}$. The linearization of the nonlinear system around the $i$-th operating point is given by
\begin{equation}
    \dot{x} \approx g(x_i^*, u_i^*) + A_i(x-x_i^*) + B_i(u-u_i^*),
\end{equation}
where $A_i =\frac{\partial g}{\partial x} |_{(x_i^*, u_i^*)} \in \mathbb{R}^{n \times n}$ and $B_i =\frac{\partial g}{\partial u} |_{(x_i^*, u_i^*)} \in \mathbb{R}^{n \times p}$.
Neglecting the high order terms gives the linearized subsystem associated with region $\Omega_i$
\begin{align}
\begin{split}
    \dot{x} \approx A_i x + B_i u + f_i, \quad x \in \Omega_i
\end{split}
\end{align}
with $f_i =g(x_i^*, u_i^*) - A_i x_i^* - B_i u_i^* \in \mathbb{R}^{n}$. To characterize in which region the state-input vector locates, we define the following indicator function
\begin{equation}
    \chi_i(t)=
    \begin{cases}
        1, & \quad \text{if}(x(t), u(t)) \in \Omega_i\\
        0, & \quad \text{otherwise}
    \end{cases}
\end{equation}
Since the regions $\{\Omega_i\}_{i=1}^s$ have no overlaps, we have $\sum_{i=1}^s \chi_i=1$ and $\prod_{i=1}^s \chi_i =0$. Thus, the PWA system can be written as
\begin{align}
\begin{split}
\label{eqn: plant_ss}
    \dot{x}(t)=A(t) x(t) + B(t) u(t) + f(t)
\end{split}
\end{align}
with $A(t)=\sum_{i=1}^s \chi_i(t) A_i$, $B(t)=\sum_{i=1}^s \chi_i(t) B_i$ and $f(t)=\sum_{i=1}^s \chi_i(t) f_i$. 

In this paper, the reference system is also chosen to be a PWA model, which provides more design flexibility for the user. Without loss of generality, we let the reference PWA system (\ref{eqn: ref sys}) and the controlled PWA system (\ref{eqn: plant_ss}) have the same region partitions and therefore, the same indicator functions. The PWA reference system is given by
\begin{equation}
    \dot{x}_m(t)=A_m(t) x_m(t) + B_m(t) r(t) + f_m(t),
    \label{eqn: ref sys}
\end{equation}
where $x_m \in \mathbb{R}^n$ and $r \in \mathbb{R}^p$ denote the state and input of the reference system, $A_m(t) = \sum_{i=1}^s \chi_i(t) A_{mi}$, $B_m(t) = \sum_{i=1}^s \chi_i(t) B_{mi}$, $f_m(t) = \sum_{i=1}^s \chi_i(t) f_{mi}$ with $A_{mi} \in \mathbb{R}^{n \times n}, B_{mi} \in \mathbb{R}^{n \times p}, f_{mi} \in \mathbb{R}^{n}, i \in \mathcal{I}$ being the parameters of the reference system.  $A_{mi}$ are Hurwitz matrices and there exists a set of positive definite matrices $P_i$ and $Q_i \in \mathbb{R}^{n \times n}, i \in \mathcal{I}$ such that
\begin{equation}
\label{eqn: lyap_eq}
    A_{mi}^T P_i + P_i A_{mi} = - Q_i, \quad \forall i\in \mathcal{I}
\end{equation}

For each subsystem, a set of controller gains is utilized. Let $K_{xi}^* \in \mathbb{R}^{p \times n}, K_{ri}^* \in \mathbb{R}^{p \times p}, K_{fi}^* \in \mathbb{R}^{p}, i\in \mathcal{I}$ denote the nominal controller gains for the $i$-th subsystem of (\ref{eqn: plant_ss}). The controller gains and the system parameters switch synchronously and therefore, the controller takes the form
\begin{equation}
\label{eqn: controller_nominal}
u(t)=K_{x}^* x(t)+ K_r^* r(t) + K_f^*,
\end{equation}
where $K_{x}^*(t)=\sum_{i=1}^s \chi_i(t) K_{xi}^*$, $K_{r}^*(t)=\sum_{i=1}^s \chi_i(t) K_{ri}^*$, $K_{f}^*(t)=\sum_{i=1}^s \chi_i(t) K_{fi}^*$. Taking (\ref{eqn: controller_nominal}) into (\ref{eqn: plant_ss}) yields the closed-loop system, which should have the same behavior as the reference system. That gives the matching equations
\begin{align}
    \begin{split}
        A_{mi}=A_i+B_i K_{xi}^*,\quad
        B_{mi}=B_i K_{ri}^*,\quad
        f_{mi}=f_i+B_i K_{fi}^*,\quad \forall i\in\mathcal{I}
    \end{split}  
\end{align}
Since $A_i, B_i, f_i$ are unknown, the nominal controller gains $K_{xi}^*, K_{ri}^*, K_{fi}^*$ are not available. Let $K_{xi}(t)\in \mathbb{R}^{p \times n}, K_{ri}(t) \in \mathbb{R}^{p \times p}, K_{fi}(t) \in \mathbb{R}^p$ be the estimates of $K_{xi}^*, K_{ri}^*, K_{fi}^*$ and we introduce the following adaptive controller
\begin{equation}
\label{eqn: controller_adaptive}
u(t)=K_{x}(t) x(t)+ K_r(t) r(t) + K_f(t)
\end{equation}
with  $K_{x}(t)=\sum_{i=1}^s \chi_i(t) K_{xi}(t)$, $K_{r}(t)=\sum_{i=1}^s \chi_i(t) K_{ri}(t)$ and $K_{f}(t)=\sum_{i=1}^s \chi_i(t) K_{fi}(t)$. Inserting (\ref{eqn: controller_adaptive}) into the controlled PWA system (\ref{eqn: plant_ss}) and defining the state tracking error $e(t)=x(t)-x_m(t)$, we have
\begin{equation}
\label{eqn: error_dyn}
    \dot{e}=A_m e + \sum_{i=1}^s \chi_i B_i (\Tilde{K}_{xi} x + \Tilde{K}_{ri} r + \Tilde{K}_{fi}),
\end{equation}
where $\Tilde{K}_{xi}=K_{xi}-K_{xi}^*, \Tilde{K}_{ri}=K_{ri}-K_{ri}^*, \Tilde{K}_{fi}=K_{fi}-K_{fi}^*$.

We define $t_0$ to be the initial time instant and the set $\{t_1, t_2, \cdots, t_k, \cdots | k \in \mathbb{N}^+\}$ to be the switching time instants.

In this paper, we would like to design an adaptive controller for PWA systems such that the norm of the state tracking error $e$ is enforced within a predefined performance bound such that the closed-loop system has performance guarantees. The performance bound can be formulated by a performance function $\rho: \mathbb{R}^+ \to \mathbb{R}^+$, a smooth and decreasing function satisfying $\lim_{t\to \infty} \rho(t)=\rho_\infty >0$. We adopt the following commonly used performance function \cite{bechlioulis2008robust}
\begin{equation}
    \label{eqn: performance_func}
    \rho(t)=(\rho_0-\rho_\infty)\mathrm{e}^{-l(t-t_0)}+\rho_\infty,
\end{equation}
where $\rho_0, \rho_\infty, l \in \mathbb{R}^+$ and $\rho_0 > \rho_\infty$. We can see that $\rho(t)$ is smooth and decreasing with $\rho(t=t_0)=\rho_0$ and $\rho(t\to\infty)=\rho_\infty$. The performance guarantee to be satisfied can be formulated as
\begin{equation}
    \label{eqn: error_constraint}
    \|e(t)\|_{P} < \rho(t),
\end{equation}
where $\|e(t)\|_{P}$ is defined to be the weighted Euclidean norm of $e(t)$ with the weighting matrix $P$, i.e., $\|e(t)\|_{P}=(e^T(t) P e(t))^{\frac{1}{2}}$. $\|e(t)\|_{P}$ serves as a performance measure reflecting the difference between the state of the controlled system and the reference system. $P$ is equal to $P_i$ if subsystem $i$ is activated, i.e., $P=\sum_{i=1}^s \chi_i(t) P_i$. So $\|e(t)\|_{P}$ and the system parameters switch synchronously. 
\begin{remark}
Note that defining a switching performance measure $\|e(t)\|_P$ will not make our approach restrictive. If a global performance measure is desired, i.e., $\|e(t)\|_S < \rho^{*}(t)$ ($S \in \mathbb{R}^{n \times n}$ being constant and positve definite) must hold for every subsystem, then we could choose $P_i, i\in \mathcal{I}$ matrices such that
\begin{equation}
    \|e\|_S \leq \min_{i\in\mathcal{I}}\sqrt{\frac{\lambda_{\mathrm{min}}(P_i)}{\lambda_\mathrm{max}(S)}}\|e\|_P.
\end{equation}
We obtain $\|e(t)\|_S < \rho^*(t)$ if we can make $\|e\|_P \leq  \sqrt{\frac{\lambda_\mathrm{max}(S)}{\min_{i\in\mathcal{I}}\lambda_{\mathrm{min}}(P_i)}} \rho^{*}(t)\triangleq \rho(t)$. This bring us back to the form (\ref{eqn: error_constraint})
\end{remark}
The problem to be studied in this paper is formulated as follows: 
\begin{problem}
\label{prl: controller}
Given a performance function (\ref{eqn: performance_func}), a reference model (\ref{eqn: ref sys}) and a PWA system (\ref{eqn: plant_ss}) with unknown subsystem parameters $A_i, B_i, f_i$ and known regions $\Omega_i$, design an adaptive control law $u(t)$ such that the state $x(t)$ of (\ref{eqn: plant_ss}) tracks the state $x_m(t)$ of (\ref{eqn: ref sys}) with the tracking error $e(t)$ satisfying the performance guarantee (\ref{eqn: error_constraint}).
\end{problem}

\section{Adaptive Control Design}
\label{sec: method}
In this section, we propose the adaptive controller and adaptation laws of the controller gains to solve the given problem in the disturbance-free case. First, we introduce the auxiliary performance bound and explain the solution concept. Then the proposed adaptation laws are presented, which is followed by the stability analysis of the closed-loop system.

\subsection{Auxiliary Performance Bound}
\label{sec: method_1}
We define a generalized restricted potential function (barrier function) $\phi: \mathbb{R}^+ \to \mathbb{R}^+$ on the set $\mathcal{D}_{\theta}\triangleq\{e\: |\: \|e\|_{P}\in [0,\theta)\}$
\begin{equation}
\label{eqn: phi_def}
    \phi(\|e\|_{P})=\frac{\|e\|_{P}^2}{\theta^2(t)-\|e\|_{P}^2}, \quad \|e\|_{P}<\theta(t).
\end{equation}
Suppose that $\|e(t_0)\|_P <\rho(t_0)$, the set-theoretic MRAC approach for linear systems \cite{arabi2019set} suggests specifying the barrier $\theta$ to be $\rho(t)$ and designing the adaptation laws such that $\phi(\|e\|_{P})$ is bounded $\forall t\in [t_0,\infty)$, then it can be obtained that $\|e(t)\|_P <\rho(t), \forall t\in [t_0,\infty)$.

The difficulty in switched systems is that $P=\sum_{i=1}^s\chi_i(t)P_i$ leads to the jumps of $\|e(t)\|_P$ at switching instants. Suppose $\chi_i(t)=1$ for $t \in [t_{k-1},t_k)$ and $\chi_j(t)=1$ for $t \in [t_{k},t_{k+1})$ for $i\neq j, i,j \in \mathcal{I}$, we have
\begin{align}
    \begin{split}
        \|e(t_k)\|_P^2=e^T(t_k)P_j e(t_k)
        \leq \lambda_{\text{max}}(P_j)\|e(t_k)\|^2
        \leq \frac{\lambda_{\text{max}}(P_j)}{\lambda_{\text{min}}(P_i)}\|e(t_k^-)\|_{P}^2,
    \end{split}
\end{align}
which may result in $\|e(t_k)\|_P > \rho(t_k)$ for $\frac{\lambda_{\text{max}}(P_j)}{\lambda_{\text{min}}(P_i)}>1$ and $\|e(t_k^-)\|_P < \rho(t_k^-)$. This further makes the barrier function $\phi(\|e\|_P)$ invalid. We call this \textit{barrier transgression} problem.

To overcome this problem, our idea is to introduce an auxiliary performance bound, denoted by $\epsilon(t)$, which decays faster than the user-defined performance bound $\rho(t)$. $\epsilon(t)$ is reset at each switching instant such that $\|e(t_k)\|_P < \epsilon(t_k)$ for $k\in\mathbb{N}^+$. If the adaptive controller ensures $\|e\|_P < \epsilon(t)$ and if $\epsilon(t)$ is designed such that $\epsilon(t) < \rho(t)$ for $t \in [t_0,\infty)$, then the control objective (\ref{eqn: error_constraint}) is achieved.

We propose the auxiliary performance bound $\epsilon(t)$ generated by the following dynamics
\begin{equation}
\label{eqn: eps_dyn}
    \dot{\epsilon}(t)=-h \epsilon(t) + g, \quad \epsilon(t_0) \in (\frac{g}{h},\rho_0), \quad  \epsilon(t_k)=G(\epsilon(t_k^-)),
\end{equation}
with $h, g \in \mathbb{R}^+$. $G: \mathbb{R}^+ \to \mathbb{R}^+$ is a state reset map. It resets the value of $\epsilon$ at each switching instant. Note that $\epsilon$ shares the same switching instants with the controlled PWA system $t_k, k\in \mathbb{N}^+$, i.e., every time when the switch of the controlled PWA system occurs, $\epsilon$ is reset by the state reset map simultaneously. We specify the state reset map $G$ to be 
\begin{equation}
    G(\epsilon(t_k^-))=\sqrt{\mu} \epsilon(t_k^-), 
\end{equation}
for some $\mu \in \mathbb{R}^+$ and $\mu > 1$.
As stated before, $\epsilon(t)$ should be smaller than $\rho(t), \forall t \in [t_0, \infty)$. To achieve this, the state reset of $\epsilon(t)$ needs to satisfy some dwell time constraints, i.e., $\min\{{t_{k}-t_{k-1}}\} \geq \tau_D, k\in \mathbb{N}^+$ for some $\tau_D \in \mathbb{R}^+$. We have the following lemma:

\begin{lemma}
\label{thm: eps_rho}
Given the performance function (\ref{eqn: performance_func}) and the auxiliary performance bound (\ref{eqn: eps_dyn}), if $h>l$, $\rho_\infty>\sqrt{\mu}\frac{g}{h}$ and if the dwell time of $\epsilon(t)$ satisfies
\begin{equation}
\label{eqn: condition_tau}
    \tau_D \geq \frac{1}{h-l}\ln{\frac{\sqrt{\mu} \rho_\infty-\frac{g}{h}\sqrt{\mu}}{\rho_\infty-\frac{g}{h}\sqrt{\mu}}}
\end{equation}
for some $\mu > 1$, then the following inequality holds
\begin{equation}
    \frac{g}{h}\leq \epsilon(t) < \rho(t),\quad \forall t\in[t_0,\infty)
\end{equation}
\end{lemma}
The proof of Lemma \ref{thm: eps_rho} can be seen in \ref{apd: lemma_1}.

Since $\epsilon$, the reference system (\ref{eqn: ref sys}) and the closed-loop system share the same switching signal, the first question to ask is, if the reference system is stable with the dwell time constraint (\ref{eqn: condition_tau})? This is answered by the following lemma.

\begin{lemma}
\label{thm: ref_stability}
The reference system (\ref{eqn: ref sys}) satisfying (\ref{eqn: lyap_eq}) is stable with the dwell time constraint (\ref{eqn: condition_tau}).
\end{lemma}
The proof of Lemma \ref{thm: ref_stability} can be seen in \ref{apd: lemma_2}.

\subsection{Adaptation Laws}
Based on the auxiliary performance bound proposed in Section \ref{sec: method_1}, we define the following generalized restricted potential function (barrier function) $\phi: \mathbb{R}^+ \to \mathbb{R}^+$ 
\begin{equation}
\label{eqn: phi_def}
    \phi(\|e\|_{P})=\frac{\|e\|_{P}^2}{\epsilon^2(t)-\|e\|_{P}^2}, \quad \|e\|_{P}<\epsilon(t)
\end{equation}
with $P=\sum_{i=1}^s \chi_i(t) P_i$. Since $\|e\|_{P}^2$ and $\epsilon^2(t)$ are piecewise continuous and piecewise differentiable, the partial derivative of $\phi$ with respect to $\|e\|_P^2$ over the time interval $[t_k, t_{k+1})$ takes the form $\phi_d \triangleq {\partial \phi}/{\partial \|e\|_P^2}=\epsilon^2(t)/(\epsilon^2(t)-\|e\|_{P}^2)^2>0$. $\phi$ and $\phi_d$ have the property that
$2\phi_d(\|e\|_{P})\|e\|_{P}^2-\phi>0$.

The adaptation laws of the estimated controller gains are given as
\begin{align}
\begin{split}
\label{eqn: adaptation_law}
    \dot{K}_{xi}&=-\chi_i \phi_d S^T B_{mi}^T P_i e x^T\\
    \dot{K}_{ri}&=-\chi_i \phi_d S^T B_{mi}^T P_i e r^T\\
    \dot{K}_{fi}&=-\chi_i \phi_d S^T B_{mi}^T P_i e
\end{split}
\end{align}
where $S_i \in \mathbb{R}^{p\times p}$ is a matrix such that there exists a symmetric and positive definite matrix $M_i \in \mathbb{R}^{p\times p}$ with $(K_{ri}^*S_i)^{-1}=M_i$. Here we make the usual assumption in adaptive control \cite{tao2014multivariable} that $S_i$ is known. The use of the indicator functions $\chi_i(t)$ in the adaptation laws (\ref{eqn: adaptation_law}) implies that the controller gains associated with a certain subsystem are updated only when this subsystem is activated. Their adaptation terminates and their values stay unchanged during the inactive phase of the corresponding subsystem.

\subsection{Stability Analysis}
The tracking performance and the stability of the closed-loop system are summarized in the following theorem.
\begin{theorem}
\label{thm: direct_stability}
Given the reference PWA system (\ref{eqn: ref sys}) and the predefined performance function (\ref{eqn: performance_func}), let the PWA system (\ref{eqn: plant_ss}) with known regions $\Omega_i, i\in \mathcal{I}$ and unknown subsystem parameters $A_i, B_i, f_i, i\in \mathcal{I}$ be controlled by the feedback controller (\ref{eqn: controller_adaptive}) with the adaptation laws (\ref{eqn: adaptation_law}). Let the initial state of $\epsilon$ satisfies $\|e(t_0)\|_P<\epsilon(t_0)$. The closed-loop system is stable and the state tracking error $e(t)$ satisfies the prescribed performance guarantees (\ref{eqn: error_constraint}) if the time constant $h$ in (\ref{eqn: eps_dyn}) satisfies
\begin{equation}
\label{eqn: condition_h}
    h< \frac{1}{2}\min_{i \in \mathcal{I}}\frac{\lambda_{\text{min}}(Q_i)}{\lambda_{\text{max}}(P_i)}
\end{equation}
and if the switching signal of the controlled PWA system obeys the dwell time constraint $\tau_D$ in (\ref{eqn: condition_tau}) with
\begin{equation}
    \mu \triangleq \max_{i,j \in \mathcal{I}}\frac{\lambda_{\text{max}}(P_i)}{\lambda_{\text{min}}(P_j)}.
\end{equation}
\end{theorem}
\begin{proof}
Consider the following Lyapunov function
\begin{equation}
\label{eqn: V}
    V=\phi(\|e\|_P)+\underbrace{\sum_{i=1}^s (\mathrm{tr}(\tilde{K}_{xi}^T M_{si} \tilde{K}_{xi}) + \mathrm{tr}(\tilde{K}_{ri}^T M_{si} \tilde{K}_{ri}) + \tilde{K}_{fi}^T M_{si} \tilde{K}_{fi})}_{\triangleq V_{\theta}}
\end{equation}
The stability analysis can be divided into two phases:

\textit{phase 1:} $t \in [t_{k-1},t_k), k\in \mathbb{N}^+$

$V$ is continuous in the intervals between two successive switches. Without loss of generality, we suppose that the $i$-th subsystem is activated for $t\in[t_{k-1},t_k)$ and $e(t_{k-1})$ satisfies $\|e(t_{k-1})\|_{P_i} < \epsilon(t_{k-1})$. The time-derivative of $V$ in $[t_{k-1},t_k)$ is given by
\begin{equation}
\label{eqn: V_dot}
\dot{V}=\dot{\phi}(\|e\|_{P_i})+  2\sum_{i=1}^s (\mathrm{tr}(\tilde{K}_{xi}^T M_{si} \dot{\tilde{K}}_{xi}) + \mathrm{tr}(\tilde{K}_{ri}^T M_{si} \dot{\tilde{K}}_{ri}) + \tilde{K}_{fi}^T M_{si} \dot{\tilde{K}}_{fi})
\end{equation}
First, we simplify the second term of $\dot{V}$. Taking the adaptation laws (\ref{eqn: adaptation_law}) into the first summand of the second term of $\dot{V}$ gives
\begin{align}
    \begin{split}
        \mathrm{tr}(\tilde{K}_{xi}^T M_{si} \dot{\tilde{K}}_{xi})
        =-\chi_i \phi_d \mathrm{tr}(\tilde{K}_{xi}^T M_{si} S^T B_{mi}^T P_i e x^T)
    \end{split}
\end{align}
Since $(K_{ri}^*S_i)^{-1}=M_i$ and $B_i K_{ri}^*=B_{mi}$, we have $M_{si} S^T B_{mi}^T=M_{si} S^T (B_i K_{ri}^*)^T=M_{si} M_{si}^{-1} B_{i}^T=B_i^T$, which further gives
\begin{align}
    \begin{split}
        \mathrm{tr}(\tilde{K}_{xi}^T M_{si} \dot{\tilde{K}}_{xi})
        &=-\chi_i \phi_d \mathrm{tr}(\tilde{K}_{xi}^T B_i^T P_i e x^T)\\
        &=-\chi_i \phi_d \mathrm{tr}(x e^T P_i B_i \tilde{K}_{xi})\\
        &=-\chi_i \phi_d \mathrm{tr}(e^T P_i B_i \tilde{K}_{xi} x)\\
        &=-\chi_i \phi_d e^T P_i B_i \tilde{K}_{xi} x
    \end{split}
\end{align}
Doing the same simplification for $\mathrm{tr}(\tilde{K}_{ri}^T M_{si} \dot{\tilde{K}}_{ri})$ and $\tilde{K}_{fi}^T M_{si} \dot{\tilde{K}}_{fi}$ we have
\begin{align}
    \begin{split}
        &2\sum_{i=1}^s (\mathrm{tr}(\tilde{K}_{xi}^T M_{si} \dot{\tilde{K}}_{xi}) + \mathrm{tr}(\tilde{K}_{ri}^T M_{si} \dot{\tilde{K}}_{ri}) + \tilde{K}_{fi}^T M_{si} \dot{\tilde{K}}_{fi})\\
        =-&2\sum_{i=1}^s \chi_i \phi_d e^T P_i B_i (\tilde{K}_{xi} x+\tilde{K}_{ri} r+\tilde{K}_{fi})
    \end{split}
\end{align}
$\dot{\phi}$ can be further simplified as 
\begin{align}
    \begin{split}
    \label{eqn: phi_dot_tmp1}
        \dot{\phi}=\frac{\partial \phi}{\partial \|e\|_{P_i}^2}\frac{\mathrm{d} \|e\|_{P_i}^2}{\mathrm{d} t}+\frac{\partial \phi}{\partial \epsilon}\dot{\epsilon}
        =2\phi_d(\|e\|_{P_i})e^T P_i \dot{e}+\frac{\partial \phi}{\partial \epsilon}\dot{\epsilon}
    \end{split}
\end{align}
Substituting $\dot{e}$ with (\ref{eqn: error_constraint}) yields
\begin{align}
    \begin{split}
    \label{eqn: phi_dot_tmp2}
        \dot{\phi}
        &=\phi_d(e^T (A_m^T P_i +P_i A_m) e + 2 e^T P_i\sum_{i=1}^s \chi_i B_i (\Tilde{K}_{xi} x + \Tilde{K}_{ri} r + \Tilde{K}_{fi}))
        +\frac{\partial \phi}{\partial \epsilon}\dot{\epsilon}\\
        &=-\phi_d e^T Q_i e + 2\sum_{i=1}^s \chi_i \phi_d e^T P_iB_i (\Tilde{K}_{xi} x + \Tilde{K}_{ri} r + \Tilde{K}_{fi})
        +\frac{\partial \phi}{\partial \epsilon}\dot{\epsilon}.
    \end{split}
\end{align}
Therefore, $\dot{V}$ can be simplified as
\begin{equation}
\label{eqn: V_dot_simplified}
    \dot{V}=-\phi_d e^T Q_i e
        +\frac{\partial \phi}{\partial \epsilon}\dot{\epsilon}
\end{equation}
with
\begin{align}
    \begin{split}
        \frac{\partial \phi}{\partial \epsilon}\dot{\epsilon}=\frac{-2\epsilon\|e\|_{P_i}^2}{(\epsilon^2-\|e\|_{P_i}^2)^2}\dot{\epsilon}=-2\phi_d(\|e\|_{P_i})\|e\|_{P_i}^2\frac{\dot{\epsilon}}{\epsilon} \leq 2\phi_d(\|e\|_{P_i})\|e\|_{P_i}^2\frac{|\dot{\epsilon}|}{\epsilon}.
    \end{split}
\end{align}
Invoking Lemma \ref{thm: eps_rho}, we have $\epsilon(t) \geq \frac{g}{h}, \forall t \in [t_0,\infty)$ and therefore, 
\begin{equation}
    \frac{|\dot{\epsilon}|}{\epsilon} = \frac{h\epsilon-g}{\epsilon}=h-\frac{g}{\epsilon} \leq h,
\end{equation}
which leads to
\begin{align}
    \begin{split}
        \frac{\partial \phi}{\partial \epsilon}\dot{\epsilon} \leq 2 h\phi_d(\|e\|_{P_i})\|e\|_{P_i}^2.
    \end{split}
\end{align}
Taking this into (\ref{eqn: V_dot_simplified}) yields
\begin{align}
    \begin{split}
        \dot{V}
        &\leq -\phi_d \|e\|_2^2 \lambda_{\text{min}}(Q_i)
        +2 h\phi_d\|e\|_2^2 \lambda_{\text{max}}(P_i)\\
        &= -\phi_d \|e\|_2^2 (\lambda_{\text{min}}(Q_i)-2 h\lambda_{\text{max}}(P_i)).
    \end{split}
\end{align}
From the condition (\ref{eqn: condition_h}) it follows $\lambda_{\text{min}}(Q_i)-2 h\lambda_{\text{max}}(P_i)>0$, which together with the property $2\phi_d(\|e\|_{P})\|e\|_{P}^2-\phi>0$ gives
\begin{align}
    \begin{split}
        \dot{V}
        &\leq -\frac{\lambda_{\text{min}}(Q_i)-2 h\lambda_{\text{max}}(P_i)}{2\lambda_{\text{max}}(P_i)}\phi\leq 0.
    \end{split}
\end{align}
The fact $\dot{V}\leq 0$ in intervals $[t_{k-1},t_{k}), k\in \mathbb{N}^+$ implies that the Lyapunov function decreases between two consecutive switches. $\phi$ and $\phi_d$ are bounded in $[t_{k-1},t_k)$. Since $\|e(t_{k-1})\|_{P_i}<\epsilon(t_{k-1})$, we have $\|e(t)\|_{P_i}<\epsilon(t)$ for $\forall t\in [t_{k-1},t_k)$.

\textit{phase 2: jump at switch instants $t_k, k\in \mathbb{N}^+$}

Now we analyse the behavior of the Lyapunov function at the switching time instants. Suppose that $i$-th subsystem is activated in $[t_{k-1},t_k)$ and $j$-th subsystem is activated in $[t_k,t_{k+1})$, where $i, j \in \mathcal{I}, i\neq j$. From the adaptation laws of the estimated controller gains (\ref{eqn: adaptation_law}), we see that the estimated controller gains are continuous and therefore $\tilde{K}_{xi}(t_k)=\tilde{K}_{xi}(t_k^-)$, $\tilde{K}_{ri}(t_k)=\tilde{K}_{ri}(t_k^-)$ and $\tilde{K}_{fi}(t_k)=\tilde{K}_{fi}(t_k^-)$ for $\forall i\in\mathcal{I}$, from which it follows $V_{\theta}(t_k^-)=V_{\theta}(t_k)$. To study the relationship between $V(t_k)$ and $V(t_k^-)$, it remains to analyse $\phi(\|e(t_k)\|_P)$ and $\phi(\|e(t_k^-)\|_P)$. Since $e(t)$ is also continuous, $e(t_k)=e(t_k^-)$. This results in
\begin{align}
    \begin{split}
    \label{eqn: e_P_norm_ineq}
        \|e(t_k)\|_P^2&=e^T(t_k)P_j e(t_k)
        \leq \lambda_{\text{max}}(P_j)\|e(t_k)\|^2\\
        &\leq \frac{\lambda_{\text{max}}(P_j)}{\lambda_{\text{min}}(P_i)}e^T(t_k)P_i e(t_k)
        =\frac{\lambda_{\text{max}}(P_j)}{\lambda_{\text{min}}(P_i)}\|e(t_k^-)\|_{P}^2
        \leq \mu \|e(t_k^-)\|_{P}^2.
    \end{split}
\end{align}
From the analysis of \textit{phase 1}, we already know that $\|e(t_k^-)\|_P < \epsilon(t_k^-)$. $\epsilon$ is reset at $t_k$ and we have
\begin{equation}
    \|e(t_k)\|_P \leq \sqrt{\mu}\|e(t_k^-)\|_{P} < \sqrt{\mu} \epsilon(t_k^-) = \epsilon(t_k), 
\end{equation}
which makes the potential function $\phi(\|e(t_k)\|_P)$ also valid at $t_k$. Recalling the dynamics of $\epsilon$ (\ref{eqn: eps_dyn}) and the above inequalities (\ref{eqn: e_P_norm_ineq}), we have
\begin{align}
    \begin{split}
        \phi(\|e(t_k)\|_P)
        =\frac{\|e(t_k)\|_P^2}{\epsilon^2(t_k)-\|e(t_k)\|_P^2}
        &\leq \frac{\mu \|e(t_k^-)\|_P^2}{\epsilon^2(t_k)-\mu\|e(t_k^-)\|_P^2}\\
        &=\frac{\mu \|e(t_k^-)\|_P^2}{\mu\epsilon^2(t_k^-)-\mu\|e(t_k^-)\|_P^2}
        =\phi(\|e(t_k^-)\|_P).
    \end{split}
\end{align}
Combining the facts $\phi(\|e(t_k)\|_P) \leq \phi(\|e(t_k^-)\|_P)$ and $V_{\theta}(t_k^-)=V_{\theta}(t_k)$, we have
\begin{equation}
    V(t_k)
    =\phi(\|e(t_k)\|_P)+V_{\theta}(t_k)
    \leq \phi(\|e(t_k^-)\|_P)+V_{\theta}(t_k^-)
    =V(t_k^-).
\end{equation}
Therefore, the Lyapunov function is non-increasing at every switching time instant. This together with the fact $\dot{V} \leq 0$ in $[t_k,t_{k+1})$ for $\forall k\in \mathbb{N}^+$ implies that $V(t)$ is non-increasing for $\forall t \in [t_0, \infty)$. The discontinuity of the Lyapunov function does not introduce extra dwell time constraints.

Combining the analysis of \textit{phase 1} and \textit{phase 2}, we have $\phi,\tilde{K}_{xi},\tilde{K}_{ri},\tilde{K}_{fi} \in \mathcal{L}_\infty$ and therefore $K_{xi},K_{ri},K_{fi} \in \mathcal{L}_\infty$. Besides, $\|e(t)\|_P < \epsilon(t) \leq \rho(t)$ holds for $\forall t\ \in [t_0,\infty)$. 

Invoking Lemma \ref{thm: ref_stability} we have $x_m \in \mathcal{L}_\infty$. $x_m \in \mathcal{L}_\infty$ and $\|e(t)\|_P < \epsilon(t) \leq \rho(t)$ lead to $x \in \mathcal{L}_\infty$, which together with $r, \phi_d \in \mathcal{L}_\infty$ implies $\dot{K}_{xi},\dot{K}_{ri},\dot{K}_{fi} \in \mathcal{L}_\infty$.
\end{proof}

Theorem \ref{thm: direct_stability} shows the tracking performance and the stability of the closed-loop system under the dwell time constraints (\ref{eqn: condition_tau}). Now we study the case with arbitrary switching. For the PWA reference systems with common Lyapunov matrix $P_0$, i.e., if positive definite matrices $P$ and $Q_i, i \in \mathcal{I}$ exist such that
\begin{equation}
\label{eqn: common_P}
    A_{mi}^T P + P A_{mi} < - Q_i, \quad i\in \mathcal{I},
\end{equation}
the error metric $\|e(t)\|_P$ exhibits no jumps at the switching instants. We can construct the potential function with the user-defined performance function directly
\begin{equation}
\label{eqn: phi_def}
    \phi_0(\|e\|_{P})=\frac{\rho^2}{\rho^2(t)-\|e\|_{P}^2}, \quad \|e\|_{P}<\rho(t).
\end{equation}
\begin{corollary}
For the reference PWA system (\ref{eqn: ref sys}) with a common Lyapunov matrix $P$, if the adaptation laws
\begin{align}
\begin{split}
\label{eqn: adaptation_law_commom_P}
    \dot{K}_{xi}&=-\chi_i \phi_{d0} S^T B_{mi}^T P e x^T\\
    \dot{K}_{ri}&=-\chi_i \phi_{d0} S^T B_{mi}^T P e r^T\\
    \dot{K}_{fi}&=-\chi_i \phi_{d0} S^T B_{mi}^T P e
\end{split}
\end{align}
are used with $\phi_{0d}\triangleq\frac{\partial \phi_0}{\partial \|e\|_P^2}$, and if the decaying rate of $\rho$ satisfies
\begin{equation}
\label{eqn: condition_h_common_Lyap}
    l < \frac{1}{2}\min_{i \in \mathcal{I}}\frac{\lambda_{\text{min}}(Q_i)}{\lambda_{\text{max}}(P)},
\end{equation}
the closed-loop system is stable under arbitrary switching and the state tracking error $e(t)$ satisfies the prescribed performance guarantees (\ref{eqn: error_constraint}).
\end{corollary}
\begin{proof}
We propose the following Lyapunov function
\begin{equation}
\label{eqn: V_0}
    V=\phi_0(\|e\|_P)+\sum_{i=1}^s (\mathrm{tr}(\tilde{K}_{xi}^T M_{si} \tilde{K}_{xi}) + \mathrm{tr}(\tilde{K}_{ri}^T M_{si} \tilde{K}_{ri}) + \tilde{K}_{fi}^T M_{si} \tilde{K}_{fi}).
\end{equation}
$V$ is continuous not only within each interval $[t_k, t_{k+1}), k\in \mathbb{N}$ but also at switch instants $t_k, k\in \mathbb{N}$. So it is a common Lyapunov function. Taking its time derivative and inserting (\ref{eqn: adaptation_law_commom_P}) and (\ref{eqn: error_dyn}), we obtain
\begin{equation}
\label{eqn: V_0_dot_simplified}
    \dot{V}=-\phi_{d0} e^T (\sum_{i=1}^s \chi_i Q_i) e +\frac{\partial \phi_0}{\partial \rho}\dot{\rho}.
\end{equation}
Since $\frac{\partial \phi_0}{\partial \rho}\dot{\rho}\leq 2\phi_{d0}(\|e\|_P)\|e\|_P^2\frac{|\dot{\rho}|}{\rho}$ and $\frac{|\dot{\rho}|}{\rho}\leq l$, we have
\begin{align}
    \begin{split}
        \dot{V}
        &\leq -\phi_d \|e\|_2^2 \min_{i\in\mathcal{I}}\lambda_{\text{min}}(Q_i)
        +2l\phi_d\|e\|_2^2 \lambda_{\text{max}}(P)\\
        &= -\phi_d \|e\|_2^2 (\min_{i\in\mathcal{I}} \lambda_{\text{min}}(Q_i)-2l\lambda_{\text{max}}(P))\\
        &\leq -\frac{\min_{i\in\mathcal{I}}\lambda_{\text{min}}(Q_i)-2l\lambda_{\text{max}}(P)}{2\lambda_{\text{max}}(P)}\phi\leq 0
    \end{split}
\end{align}
given that (\ref{eqn: condition_h_common_Lyap}) holds. $\dot{V} \leq 0$ is negative semidefinite. Therefore, we have $\phi,\tilde{K}_{xi},\tilde{K}_{ri},\tilde{K}_{fi} \in \mathcal{L}_\infty$ for arbitrary switching. The boundedness of $\tilde{K}_{xi},\tilde{K}_{ri},\tilde{K}_{fi}$ implies $K_{xi},K_{ri},K_{fi} \in \mathcal{L}_\infty$. Furthermore, $\|e(t)\|_P < \epsilon(t) \leq \rho(t)$ holds for $\forall t\ \in [t_0,\infty)$. This leads to $x \in \mathcal{L}_\infty$, which together with $r, \phi_d \in \mathcal{L}_\infty$ implies that $\dot{K}_{xi},\dot{K}_{ri},\dot{K}_{fi} \in \mathcal{L}_\infty$.
\end{proof}
\begin{remark}
The classical MRAC approaches for PWL and PWA systems \cite{sang2012adaptive1,kersting2017direct} suggest using $e^T(\sum_{i=1}^s \chi_i P_i) e$ as the error-related term (the first summand) of the Lyapunov function $V$. This leads to potential increase of $V$ at switching instants. The dwell time constraints are then derived by formulating an inequality in form of $\dot{V} < -\alpha V+\beta$ for some constant $\alpha,\beta>0$ to keep $V$ exponentially decreasing in between the switches. To achieve this, the projection operator needs to be introduced (see \cite{sang2012adaptive1}) or the input signal must be PE (see \cite{kersting2017direct}) in the disturbance-free case. One key feature of our approach is that the Lyapunov function $V$ is non-increasing even at the switching instants and does not impose dwell time constraints. This omits the need of introducing projection or PE condition in the disturbance-free case. 
\end{remark}

\section{Robust Adaptive Control}
In Section \ref{sec: method}, the adaptive control approach and the stability of the closed-loop systems are studied in the disturbance-free case. Since the PWA systems are commonly used as the approximation of nonlinear systems, approximation errors exist. Besides, unmodeled dynamics and external disturbances cannot be neglected in real applications. In this section, we focus on the robust adaptive control design for PWA systems with approximation errors, unmodeled dynamics, and external disturbances, i.e., we consider 
\begin{align}
\begin{split}
\label{eqn: plant_ss_with_d}
    \dot{x}(t)=A(t) x(t) + B(t) u(t) + f(t) + d(x,u,t),
\end{split}
\end{align}
where $d(x,u,t) \in \mathbb{R}^n$ can denote the approximation error of the linearization, unmodeled dynamics or external disturbances. $d$ is continuous and its norm is upper bounded, i.e., $\|d\|_2\leq \bar{d}$, where $\bar{d}$ is known.

We propose the following robust adaptation laws
\begin{align}
\begin{split}
\label{eqn: adaptation_law_robust}
    \dot{K}_{xi}&=-\chi_i \phi_d S^T B_{mi}^T P_i e x^T + \chi_i F_{xi}\\
    \dot{K}_{ri}&=-\chi_i \phi_d S^T B_{mi}^T P_i e r^T + \chi_i F_{ri}\\
    \dot{K}_{fi}&=-\chi_i \phi_d S^T B_{mi}^T P_i e + \chi_i F_{oi}
\end{split}
\end{align}
where $F_{xi} \in \mathbb{R}^{p \times n}, F_{ri} \in \mathbb{R}^{p \times p}, F_{0i} \in \mathbb{R}^{p}$ represent the projection terms to confine the estimated controller gains $K_{xi}, K_{ri}, K_{fi}$ within some given bounds. The projection terms have no effect on the adaptation if $K_{xi}, K_{ri}, K_{fi}$ are within their bounds, otherwise, the adaptation terminates. $S_i \in \mathbb{R}^{p\times p}$ is a matrix such that there exists a diagonal and positive definite matrix $M_i \in \mathbb{R}^{p\times p}$ with $(K_{ri}^*S_i)^{-1}=M_i$. 
\begin{remark}
For the robust adaptive control design, more prior information is required compared with the disturbance-free case. For our projection-based approach, $M_i$ must be diagonal and the element-wise bounds of $K_{xi}, K_{ri}, K_{fi}$ need to be known (see also \cite{sang2011adaptive}). The leakage-based approach proposed in \cite{yuan2018robust} requires $M_i$ to be completely known because they are used in the leakage terms. Its improved version in \cite{tao2020issue} requires $\lambda_{max}(M_i^{-1})$ to satisfy some constraints associated with the leakage rates.
\end{remark}
\begin{remark}
There is another popular formulation $\dot{x}=A_p x + B_p \Lambda u$ appearing in many works inspired by aerospace applications \cite{lavretsky2011adaptive, arabi2019neuroadaptive, arabi2019set}, where $B_p$ is known and $\Lambda$ is an unknown diagonal matrix with strictly positive diagonal elements. Such arrangement of the input matrix is equivalent to our requirement that $M_i$ must be diagonal and positive definite.
\end{remark}

Besides, we assume that positive definite matrices $P_i, Q_i, i \in \mathcal{I}$ exist such that
\begin{equation}
\label{eqn: lyap_eq2}
    A_{mi}^T P_i + P_i A_{mi} + P_i< - Q_i, \quad i\in \mathcal{I}.
\end{equation}

Before we proceed with the robustness analysis, another property of the potential function, which is useful for the analysis in this paper, is given in the following lemma.
\begin{lemma}
\label{thm: phi_ineq}
For a positive constant $c \in \mathbb{R}^+$ and $c < \mathrm{min}_{t} \epsilon^2(t)$, the function $\phi(\|e\|_P)$ defined in (\ref{eqn: phi_def}) and its partial derivative $\phi_d$ with respect to $\|e\|_P$ satisfy
\begin{enumerate}
    \item[(1)] $2\phi_d \cdot (\|e\|_P^2-c)-\phi>0$ for $\zeta <\|e\|_P^2 <\epsilon^2$\\
    \item[(2)] $2\phi_d \cdot (\|e\|_P^2-c)-\phi\leq 0$ for $\|e\|_P^2 \leq \zeta$
\end{enumerate}
with $\zeta \triangleq \frac{-\epsilon^2+\sqrt{\epsilon^4+4\epsilon^2 c}}{2}$.
\end{lemma}
\begin{proof}
From the definition of $\phi$ given in (\ref{eqn: phi_def}) we have
\begin{align}
    \begin{split}
    \label{eqn: lemma2_tmp}
        2\phi_d\cdot (\|e\|_P^2-c)-\phi
        &=\frac{\|e\|_P^4+\epsilon^2 \|e\|_P^2-2c\epsilon^2}{(\epsilon^2-\|e\|_P^2)^2}.
    \end{split}
\end{align}
The denominator of (\ref{eqn: lemma2_tmp}) is positive and the sign of $2\phi_d \cdot (\|e\|_P^2-c)-\phi$ is determined by the numerator, which can be viewed as a quadratic function $f(z)=z^2+\epsilon^2 z-2c\epsilon^2$ with $z=\|e\|_P^2$. We have $f(z)\leq 0$ for $z \in [\frac{-\epsilon^2-\sqrt{\epsilon^4+4\epsilon^2 c}}{2},\frac{-\epsilon^2+\sqrt{\epsilon^4+4\epsilon^2 c}}{2}]$ and $f(z) > 0$ otherwise. Since $\phi, \phi_d$ are defined over $\|e\|_P^2 \in [0, \epsilon^2)$ and $\frac{-\epsilon^2-\sqrt{\epsilon^4+4\epsilon^2 c}}{2}<0$, it can be obtained that $2\phi_d(\|e\|_P^2-c)-\phi>0$ for $\zeta <\|e\|_P^2 <\epsilon^2$ and $2\phi_d(\|e\|_P^2-c)-\phi\leq 0$ for $\|e\|_P^2 \leq \zeta$ with $\zeta=\frac{-\epsilon^2+\sqrt{\epsilon^4+4\epsilon^2 c}}{2}$.
\end{proof}
\begin{theorem}
\label{thm: robust_stability}
Given the reference PWA system (\ref{eqn: ref sys}) and the predefined performance function (\ref{eqn: performance_func}), let the PWA system (\ref{eqn: plant_ss}) with known regions $\Omega_i, i\in \mathcal{I}$ and unknown subsystem parameters $A_i, B_i, f_i, i\in \mathcal{I}$ be controlled by the feedback controller (\ref{eqn: controller_adaptive}) with the adaptation laws (\ref{eqn: adaptation_law_robust}). Let the initial state of $\epsilon$ satisfies $\|e(t_0)\|_P<\epsilon(t_0)$. The closed-loop system is stable and the state tracking error $e(t)$ satisfies the prescribed performance guarantees (\ref{eqn: error_constraint}) if the time constant $h$ in (\ref{eqn: eps_dyn}) satisfies 
\begin{align}
\begin{split}
\label{eqn: condition_h_2}
    h< 
    \frac{1}{2}\min_{i \in \mathcal{I}}\frac{\lambda_{\text{min}}(Q_i)}{\lambda_{\text{max}}(P_i)},\\
    \max_{i \in \mathcal{I}} \frac{\lambda_{\text{max}}(P_i)\bar{d}}{\sqrt{\lambda_{\mathrm{min}}(Q_i)-2h\lambda_{\mathrm{max}}(P_i)}}
    < \frac{h}{g},
\end{split}
\end{align}
and if the switching signal of the controlled PWA system obeys the dwell time constraint $\tau_D$ in (\ref{eqn: condition_tau}) with
\begin{equation}
    \mu = \max_{i,j \in \mathcal{I}}\frac{\lambda_{\text{max}}(P_i)}{\lambda_{\text{min}}(P_j)}.
\end{equation}
\end{theorem}
\begin{proof}
We propose the same Lyapunov function as (\ref{eqn: V}). The stability analysis can also be divided into two phases as the one in Theorem \ref{thm: direct_stability}.

\textit{phase 1: $t \in [t_{k-1}, t_k), k\in\mathbb{N}^+$}

Suppose that $i$-th subsystem is activated for $[t_{k-1},t_k)$, the time-derivative of $V$ in $[t_{k-1},t_k)$ is the same as shown in (\ref{eqn: V_dot}). 
Following the same steps as (\ref{eqn: phi_dot_tmp1}) and (\ref{eqn: phi_dot_tmp2}), we have
\begin{align}
    \begin{split}
    \label{eqn: phi_dot_tmp3}
        \dot{\phi}
        =-\phi_d e^T (A_{mi}^T P_i +P_i A_{mi}) e + 2 &\phi_d e^T P_iB_i (\Tilde{K}_{xi} x + \Tilde{K}_{ri} r + \Tilde{K}_{fi})\\
        &+\frac{\partial \phi}{\partial \epsilon}\dot{\epsilon}+\phi_d(e^T P_i d + d^T P_i e).
    \end{split}
\end{align}
Taking the adaptation laws (\ref{eqn: adaptation_law_robust}) into $\dot{V}$ yields
\begin{align}
\begin{split}
    \dot{V}=-\phi_d e^T (A_{mi}^T P_i +&P_i A_{mi}) e +\phi_d(e^T P_i d + d^T P_i e)+\frac{\partial \phi}{\partial \epsilon}\dot{\epsilon}\\
    &+ 2 \phi_d (\mathrm{tr}(\tilde{K}_{xi}^T M_{si} F_{xi})+\mathrm{tr}(\tilde{K}_{ri}^T M_{si} F_{ri})+\tilde{K}_{fi}^T M_{si} F_{0i})
\end{split}
\end{align}
Since $M_{si}$ is diagonal, we have
\begin{align}
\begin{split}
    &\phi_d (\mathrm{tr}(\tilde{K}_{xi}^T M_{si} F_{xi})+\mathrm{tr}(\tilde{K}_{ri}^T M_{si} F_{ri})+\tilde{K}_{fi}^T M_{si} F_{0i})\\
    =&\phi_d (\sum_{j=1}^p \sum_{l=1}^n m_{i}^{(j)}\tilde{k}_{xi}^{(jl)}f_{xi}^{(jl)}+\sum_{j=1}^p \sum_{l=1}^p m_{i}^{(j)}\tilde{k}_{ri}^{(jl)}f_{ri}^{(jl)}+\sum_{j=1}^p m_{i}^{(j)}\tilde{k}_{fi}^{(j)}f_{0i}^{(j)})
\end{split}
\end{align}
with $\tilde{K}_{xi}=[\tilde{k}_{xi}^{(jl)}]$, $\tilde{K}_{ri}=[\tilde{k}_{ri}^{(jl)}]$ and $\tilde{K}_{fi}=[\tilde{k}_{fi}^{(j)}]$. $M_{si}=\mathrm{diag}(m_i^{(1)},\cdots,m_i^{(p)})$. It can be verified that $\tilde{k}_{xi}^{(jl)}f_{xi}^{(jl)}\leq 0$, $\tilde{k}_{ri}^{(jl)}f_{ri}^{(jl)}\leq 0$ and $\tilde{k}_{fi}^{(jl)}f_{0i}^{(jl)}\leq 0$, which together with the fact that $m_i^{(j)}>0$ leads to
\begin{align}
\begin{split}
\label{eqn: V_dot_tmp}
    \dot{V}\leq -\phi_d e^T (A_{mi}^T P_i +P_i A_{mi}) e +\frac{\partial \phi}{\partial \epsilon}\dot{\epsilon}+\phi_d(e^T P_i d + d^T P_i e).
\end{split}
\end{align}
Since $P_i$ is positive definite, it can be written as $P_i=H_i H_i^T$ with $H_i$ being a nonsingular matrix. The inequality (\ref{eqn: V_dot_tmp}) can be further transformed as
\begin{align}
    \begin{split}
    \label{eqn: V_dot_tmp2}
        \dot{V}
        &\leq -\phi_d e^T (A_{mi}^T P_i +P_i A_{mi}) e +\frac{\partial \phi}{\partial \epsilon}\dot{\epsilon}+2\phi_d e^T H_i H_i^T d\\
        &\leq -\phi_d e^T (Q_i+P_i) e +\frac{\partial \phi}{\partial \epsilon}\dot{\epsilon}+\phi_d(e^T H_i H_i^T e + d^T H_i H_i^T d)\\
        &=-\phi_d e^T Q_i e +\frac{\partial \phi}{\partial \epsilon}\dot{\epsilon}+\phi_d d^T H_i H_i^T d\\
        &\leq -\phi_d \|e\|_2^2(\lambda_{\mathrm{min}}(Q_i)-2h\lambda_{\mathrm{max}}(P_i)) +\phi_d d^T P_i d\\
        &\leq -\phi_d \|e\|_2^2\kappa_i +\phi_d \lambda_{\mathrm{max}}(P_i){\bar{d}}^2
    \end{split}
\end{align}
with $\kappa_i \triangleq \lambda_{\mathrm{min}}(Q_i)-2h\lambda_{\mathrm{max}}(P_i)$. For $P_i, Q_i$ and $h$ satisfying the condition (\ref{eqn: condition_h_2}), we have $\kappa_i>0$. Further analysis can be divided into two cases: $\|e\|_P^2 > \zeta$ and $\|e\|_P^2 \leq \zeta$, where
\begin{equation}
    \zeta=\frac{-\epsilon^2 + \sqrt{\epsilon^4+4 \epsilon^2 c}}{2}
\end{equation}
with $c\triangleq\sup_{i\in \mathcal{I}}\{\frac{\lambda^2_{\mathrm{max}}(P_i)}{\kappa_i}{\bar{d}}^2\}$. From (\ref{eqn: condition_h_2}) we obtain
\begin{equation}
    \epsilon(t)^2 \geq \frac{h^2}{g^2} > 
    \max_{i \in \mathcal{I}} \frac{ \lambda^2_{\text{max}}(P_i){\bar{d}}^2}{\lambda_{\mathrm{min}}(Q_i)-2h\lambda_{\mathrm{max}}(P_i)}
    =\max_{i\in \mathcal{I}}\{\frac{\lambda^2_{\mathrm{max}}(P_i)}{\kappa_i}{\bar{d}}^2\}=c,
\end{equation}
which further leads to
\begin{equation}
    \zeta < \frac{-\epsilon^2 + \sqrt{\epsilon^4+4 \epsilon^2 \cdot \epsilon^2}}{2}=\frac{(\sqrt{5}-1)\epsilon^2}{2}<\epsilon^2.
\end{equation}

\textit{Case 1 $\|e\|_P^2 > \zeta$}: invoking Lemma \ref{thm: phi_ineq}, inequality (\ref{eqn: V_dot_tmp2}) can be further derived as
\begin{align}
    \begin{split}
    \label{eqn: V_dot_tmp3}
        \dot{V}
        \leq -\frac{\kappa_i\phi_d}{\lambda_{\mathrm{max}}(P_i)}(\|e\|_P^2-\frac{\lambda^2_{\mathrm{max}}(P_i)}{\kappa_i}{\bar{d}}^2) <-\frac{\kappa_i}{2\lambda_{\mathrm{max}}(P_i)}\phi < 0
    \end{split}
\end{align}
 
 \textit{Case 2 $\|e\|_P^2 \leq \zeta$}: defining $\kappa\triangleq \mathrm{min}_{i\in \mathcal{I}}\{\kappa_i\}$, $\alpha=\mathrm{max}_{i\in \mathcal{I}}\lambda_{\mathrm{max}}(P_i)$ and considering the property that $2\phi_d(\|e\|_{P})\|e\|_{P}^2-\phi>0$, we have
\begin{align}
    \begin{split}
        \dot{V}
        &\leq -\frac{\kappa}{2\alpha}\phi +\phi_d \alpha{\bar{d}}^2\\
        &=-\frac{\kappa}{2\alpha}(\phi+V_{\theta}) +\frac{\kappa}{2\alpha}V_{\theta}+\phi_d \alpha{\bar{d}}^2\\
        &\leq -\frac{\kappa}{2\alpha}V +\frac{\kappa}{2\alpha}V_{\theta}+\phi_{d_{\mathrm{max}}} \alpha{\bar{d}}^2
    \end{split}
\end{align}
with $\phi_{d_{\mathrm{max}}}=\mathrm{max}_{\|e\|_P^2\leq \zeta}\phi_d(\|e\|_P^2)=\phi_d(\zeta)\in \mathcal{L}_{\infty}$. $V_{\theta}$ is defined in (\ref{eqn: V}). $\tilde{K}_{xi},\tilde{K}_{ri},\tilde{K}_{fi}$ are bounded due to the utilization of the projection, which leads to $V_{\theta} \in \mathcal{L}_{\infty}$. Let the positive number $\mathcal{B} \in \mathbb{R}^+$ be defined as
\begin{equation}
    \mathcal{B}\triangleq V_{\theta}+\frac{2\phi_{d_{\mathrm{max}}} \alpha^2 {\bar{d}}^2}{\kappa}.
\end{equation}
For $V \leq \mathcal{B}$, $V$ may increase. For $V > \mathcal{B}$, we have $\dot{V} <0$ and therefore, $V$ is decreasing. Combing \textit{Case 1} and \textit{Case 2}, we know that $V$ is bounded. 

\textit{phase 2: jump at switch instants $t_k, k\in \mathbb{N}^+$}

Following the same analysis as the one shown in Theorem \ref{thm: direct_stability} we have $V(t_k) \leq V(t_k^-)$. The Lyapunov function is non-increasing at each switching instant.

Based on the analysis of phase 1 and phase 2, we can conclude that 
\begin{equation}
    V(t) \leq \mathrm{max}\{V(t_0), \mathcal{B}\}, \forall t \in [t_0,\infty),
\end{equation}
from which we obtain $\phi, \phi_d \in \mathcal{L}_\infty$. The projection leads to $\tilde{K}_{xi},\tilde{K}_{ri},\tilde{K}_{fi} \in \mathcal{L}_\infty$ and therefore $K_{xi},K_{ri},K_{fi} \in \mathcal{L}_\infty$. Besides, $\|e(t)\|_P < \epsilon(t) \leq \rho(t)$ holds for $\forall t\ \in [t_0,\infty)$. The prescribed performance guarantee (\ref{eqn: error_constraint}) is satisfied. 


With the similar steps in the proof of Lemma \ref{thm: ref_stability}, one can prove the stability of the reference system satisfying (\ref{eqn: lyap_eq}), so we have $x_m\in\mathcal{L}_\infty$. This leads to $x \in \mathcal{L}_\infty$, which together with $r, \phi_d \in \mathcal{L}_\infty$ implies $\dot{K}_{xi},\dot{K}_{ri},\dot{K}_{fi} \in \mathcal{L}_\infty$.

\begin{remark}
In works about set-theoretic MRAC \cite{arabi2018set,arabi2019set,arabi2019neuroadaptive}, the uncertainties are feed into the system through the same input matrix as the control signal. A fault tolerant set-theoretic MRAC approach proposed in \cite{xiao2019robust} also assumes the actuator fault and external disturbances to be matched, i.e., they can be compensated by designing additive terms in the control signal. Compared with these works, a distinctive feature of this paper is that the disturbance term $d$ is also allowed to be unmatched.
\end{remark}
\end{proof}

\begin{minipage}{\textwidth}
\begin{minipage}[b]{0.6\textwidth}
\begin{tikzpicture}[every node/.style={outer sep=0pt},thick,
 mass/.style={draw,thick},
 spring/.style={thick,decorate,decoration={zigzag,pre length=0.3cm,post
 length=0.3cm,segment length=6}},
 ground/.style={fill,pattern=north east lines,draw=none,minimum
 width=0.75cm,minimum height=0.3cm},
 dampic/.pic={\fill[white] (-0.1,-0.3) rectangle (0.3,0.3);
 \draw (-0.3,0.3) -| (0.3,-0.3) -- (-0.3,-0.3);
 \draw[line width=1mm] (-0.1,-0.3) -- (-0.1,0.3);}]
 
  \node[mass,minimum width=1.75cm,minimum height=1cm] (m1) {$m_1$};
  \node[mass,minimum width=1.75cm,minimum height=1cm,right=1.5cm of
  m1] (m2) {$m_2$};
  \node[left=2cm of m1,ground,minimum width=3mm,minimum height=2.5cm] (g1){};
  \draw (g1.north east) -- (g1.south east);

  \draw[spring] ([yshift=3mm]g1.east) coordinate(aux)
   -- (m1.west|-aux) node[midway,above=1mm]{$c_0$};
  \draw[spring]  (m1.east|-aux) -- (m2.west|-aux) node[midway,above=1mm]{$F_c(p_1,p_2)$};

  \draw ([yshift=-3mm]g1.east) coordinate(aux')
   -- (m1.west|-aux') pic[midway]{dampic} node[midway,below=3mm]{$d$}
     (m1.east|-aux') -- (m2.west|-aux') pic[midway]{dampic} node[midway,below=3mm]{$d$};

  \foreach \X in {1,2}  
  {\draw[thin] (m\X.north) -- ++ (0,1) coordinate[midway](aux\X);
   \draw[latex-] (aux\X) -- ++ (-0.5,0) node[above]{$F_\X$}; 
   \draw[thin,dashed] (m\X.south) -- ++ (0,-1) coordinate[pos=0.85](aux'\X);
   \draw[->] (aux'\X) -- ++ (1,0) node[midway,above]{$p_\X$}
    node[left,minimum height=7mm,minimum width=1mm] (g'\X){};
  }
\end{tikzpicture}
\captionof{figure}{The mass-spring-damper system}
\label{fig: sys}
\end{minipage}
\begin{minipage}[b]{0.4\textwidth}
\begin{tabular}{ c | c}
\hline
parameters & values\\ \hline
$m_1$ & $\SI{5}{\kilogram}$\\
$m_2$ & $\SI{1}{\kilogram}$\\
$c_0$ & $\SI[per-mode=symbol]{1}{\newton\per\metre}$\\
$d$ & $\SI[per-mode=symbol]{1}{\newton\second\per\metre}$\\
\hline
\end{tabular}
\captionof{table}{System parameters}
\label{tbl: sys}
\end{minipage}
\end{minipage}
\section{Numerical Validation}
\label{sec: validation}
In this section, the proposed MRAC approach is validated through a numerical example taken from \cite{kersting2017direct}. The system is a mass-spring-damper system, shown in the Fig. \ref{fig: sys}, where $m_1, m_2$ denote the masses, $d$ represents the damping factor. The displacement of the two spring are denoted by $p_1,p_2$, The forces operated on the masses are $F_1, F_2$, respectively. The left mass is fixed with the wall by the first spring. It has a static spring constant $c_0$. The values of the system parameters are shown in Table \ref{tbl: sys}. The two masses are connected with the second spring exhibiting a PWA stiffness characteristics
\begin{equation}
    F_c(p_1,p_2)=
    \begin{cases}
        c_1=\SI[per-mode=symbol]{10}{\newton\per\metre}, & \quad \text{if } |p_2 - p_1|\leq \SI{1}{\metre}\\
        c_2=\SI[per-mode=symbol]{1}{\newton\per\metre}, & \quad \text{if } p_2 - p_1 >\SI{1}{\metre}\\
        c_3=\SI[per-mode=symbol]{100}{\newton\per\metre}, & \quad \text{if } p_2 - p_1 <\SI{-1}{\metre}
    \end{cases}
\end{equation}

Let the state $x=[x_1,x_2,x_3,x_4]^T=[p_1,\dot{p}_1,p_2,\dot{p}_2]^T$ and the input $u=[F_1,F_2]^T$. The region partitions are given as
\begin{align*}
    \begin{split}
        &\Omega_1=\{x^T \in \mathbb{R}^4 | |x_3-x_1| \leq 1\},\\
        &\Omega_2=\{x^T \in \mathbb{R}^4 | x_3-x_1 > 1\},\\
        &\Omega_3=\{x^T \in \mathbb{R}^4 | x_3-x_1 < 1\}.
    \end{split}
\end{align*}
The system dynamics can be described by a PWA system. For example the $3$rd subsystem in the state space form is
\begin{equation}
\dot{x}=
    \begin{bmatrix}
        0 & 1 & 0 & 0\\
        -\frac{c_0+c_3}{m_1} & -\frac{2d}{m_1} & \frac{c_3}{m_1} & \frac{d}{m_1}\\
        0 & 0 & 0 & 1\\
        \frac{c_3}{m_2} & \frac{d}{m_2} & -\frac{c_3}{m_2} & -\frac{d}{m_2}
    \end{bmatrix}
    x+
    \begin{bmatrix}
        0 & 0\\
        \frac{1}{m_1} & 0\\
        0 & 0\\
        0 & \frac{1}{m_2}
    \end{bmatrix}
    u+
    \begin{bmatrix}
        0\\
        \frac{c_3-c_1}{m_1}\\
        0\\
        \frac{c_1-c_3}{m_2}
    \end{bmatrix}.
\end{equation}
The reference system is chosen as
\begin{alignat}{3}
    &A_{m1} = 
    \begin{bmatrix}
        0 & 1 & 0 & 0\\
        -25 & -10 & 0 & 0\\
        0 & 0 & 0 & 1\\
        0 & 0 & -25 & -10
    \end{bmatrix},
    \quad 
    &&B_{m1}=
    \begin{bmatrix}
        0 & 0\\
        25 & 0\\
        0 & 0\\
        0 & 25
    \end{bmatrix},
    \quad 
    &&f_{m1}=
    \begin{bmatrix}
        0\\
        0\\
        0\\
        0
    \end{bmatrix}\\
    &A_{m2} = 
    \begin{bmatrix}
        0 & 1 & 0 & 0\\
        -16 & -8 & 0 & 0\\
        0 & 0 & 0 & 1\\
        0 & 0 & -16 & -8
    \end{bmatrix},
    \quad 
    &&B_{m2}=
    \begin{bmatrix}
        0 & 0\\
        16 & 0\\
        0 & 0\\
        0 & 16
    \end{bmatrix},
    \quad 
    &&f_{m2}=
    \begin{bmatrix}
        0\\
        5\\
        0\\
        -5
    \end{bmatrix}\\
    &A_{m3} = 
    \begin{bmatrix}
        0 & 1 & 0 & 0\\
        -49 & -14 & 0 & 0\\
        0 & 0 & 0 & 1\\
        0 & 0 & -49 & -14
    \end{bmatrix},
    \quad 
    &&B_{m3}=
    \begin{bmatrix}
        0 & 0\\
        49 & 0\\
        0 & 0\\
        0 & 49
    \end{bmatrix},
    \quad 
    &&f_{m3}=
    \begin{bmatrix}
        0\\
        -10\\
        0\\
        -5
    \end{bmatrix}
\end{alignat}
Specifying 
\begin{equation}
    Q_{i}=
    \begin{bmatrix}
        100 & 10 & 0 & 0\\
        10 & 100 & 0 & 0\\
        0 & 0 & 100 & 10\\
        0 & 0 & 10 & 100
    \end{bmatrix}
    \quad \text{for } i \in \{1,2,3\},
\end{equation}
we obtain the following $P_i$ matrices
\begin{align}
    \begin{split}
        P_{1}=
        \begin{bmatrix}
            140 & 2 & 0 & 0\\
            2 & 5.2 & 0 & 0\\
            0 & 0 & 140 & 2\\
            0 & 0 & 2 & 5.2
        \end{bmatrix},
        P_{2}=
        \begin{bmatrix}
            121.25 & 3.125 & 0 & 0\\
            3.125 & 6.64 & 0 & 0\\
            0 & 0 & 121.25 & 3.125\\
            0 & 0 & 3.125 & 6.64
        \end{bmatrix},\\
        P_{3}=
        \begin{bmatrix}
            182.857 & 1.02 & 0 & 0\\
            1.02 & 3.644 & 0 & 0\\
            0 & 0 & 182.857 & 1.02\\
            0 & 0 & 1.02 & 3.644
        \end{bmatrix},
    \end{split}
\end{align}
which gives $\sqrt{\mu}=7.1$. The performance function is designed with $\rho_0=10, \rho_\infty=1.5, l=0.02$. We choose $\epsilon(t_0)=9, h=0.12$ and $g=0.01$ such that the condition (\ref{eqn: condition_h}) and further conditions stated in Lemma \ref{thm: eps_rho} hold. Let the initial values of the reference system and the controlled PWA system to be $0$. The initial values of the estimated controller gains are specified as $K_{xi}(t_0)=0.5K_{xi}^*, K_{ri}(t_0)=0.5K_{ri}^*, K_{fi}(t_0)=0.5K_{fi}^*, i\in \{1,2,3\}$. We use the input signal $r=[0.3\sin{(0.5 t+\pi)},r_2(t)]$, where 
\begin{equation}
    r_2(t)=
    \begin{cases}
        2, & \quad \text{for } kT+\SI{25}{\second}<t<kT+\SI{50}{\second}\\
        -2, & \quad \text{for } kT+\SI{75}{\second}<t<kT+\SI{100}{\second}\\
        0, & \quad \text{otherwise}
    \end{cases}
\end{equation}
with $k\in\mathbb{N}, T=\SI{100}{\second}$. 
\begin{figure}[h]
    \includegraphics[width=1\textwidth]{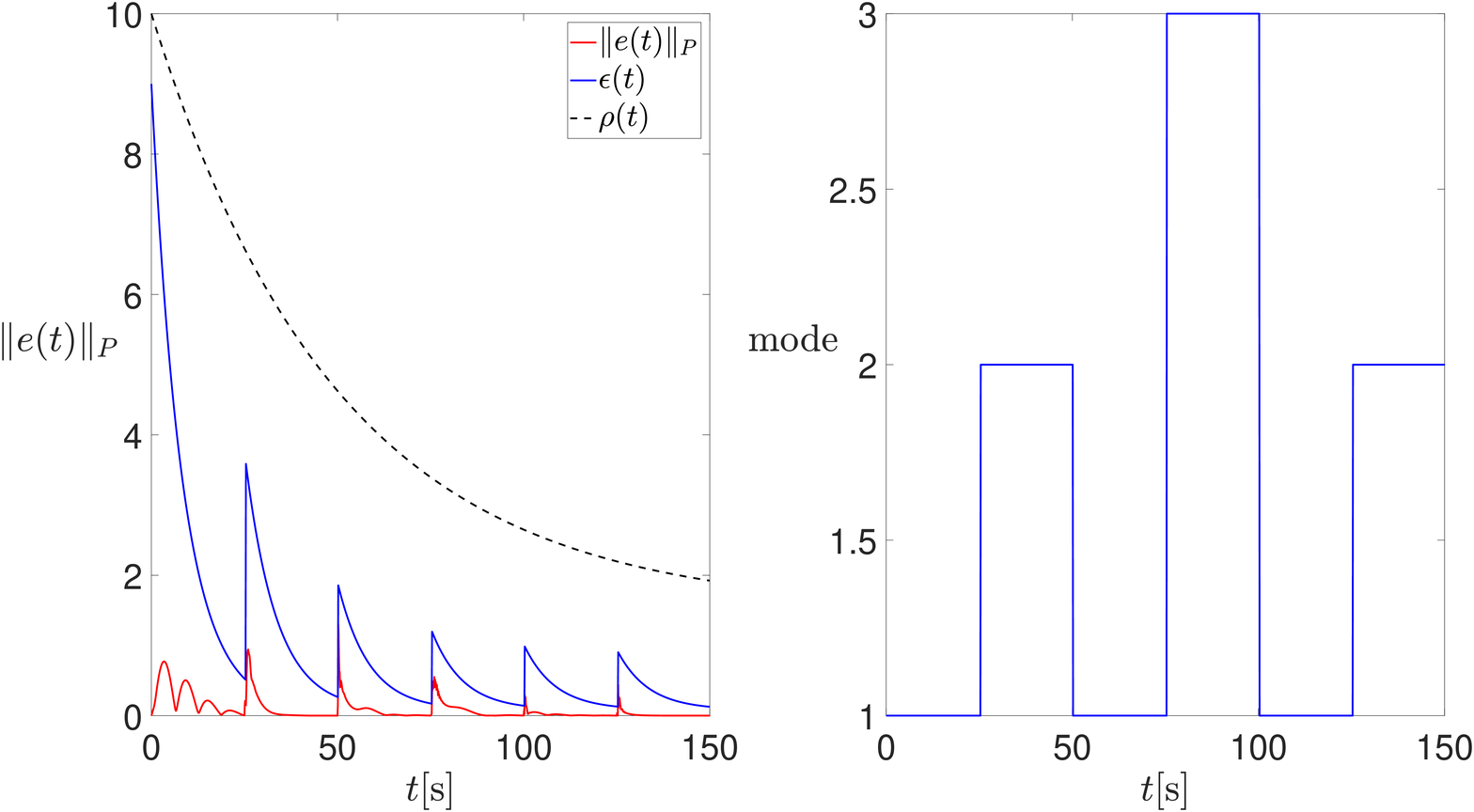}
    \caption{State tracking error, auxiliary performance signal and the predefined performance bound}
    \label{fig: e_norm}
\end{figure}

In Fig. \ref{fig: e_norm}, the prescribed performance bound $\rho(t)$, the auxiliary performance bound $\epsilon(t)$ and the weighted norm of the state tracking error $\|e(t)\|_P$ are displayed with the black dashed line, the blue solid line and the red solid line, respectively. We can see that $\|e(t)\|_P < \epsilon(t) < \rho(t)$. This guarantees the potential function $\phi(t)$ to be valid, which together with $\epsilon(t) < \rho(t)$ implies that the control objective (\ref{eqn: error_constraint}) is fulfilled. According to Theorem \ref{thm: direct_stability}, the inequality $\tau_D >\SI{24}{\second}$ should hold. We can see from the mode shown in Fig. \ref{fig: e_norm} that the dwell time constraint is satisfied.

The component-wise state tracking performance is shown in Fig. \ref{fig: state}. The red solid lines represent the state elements of the controlled PWA system and the blue dashed lines display the state elements of the reference system. Good state tracking performance can be observed.
\begin{figure}[h]
    \includegraphics[width=1\textwidth]{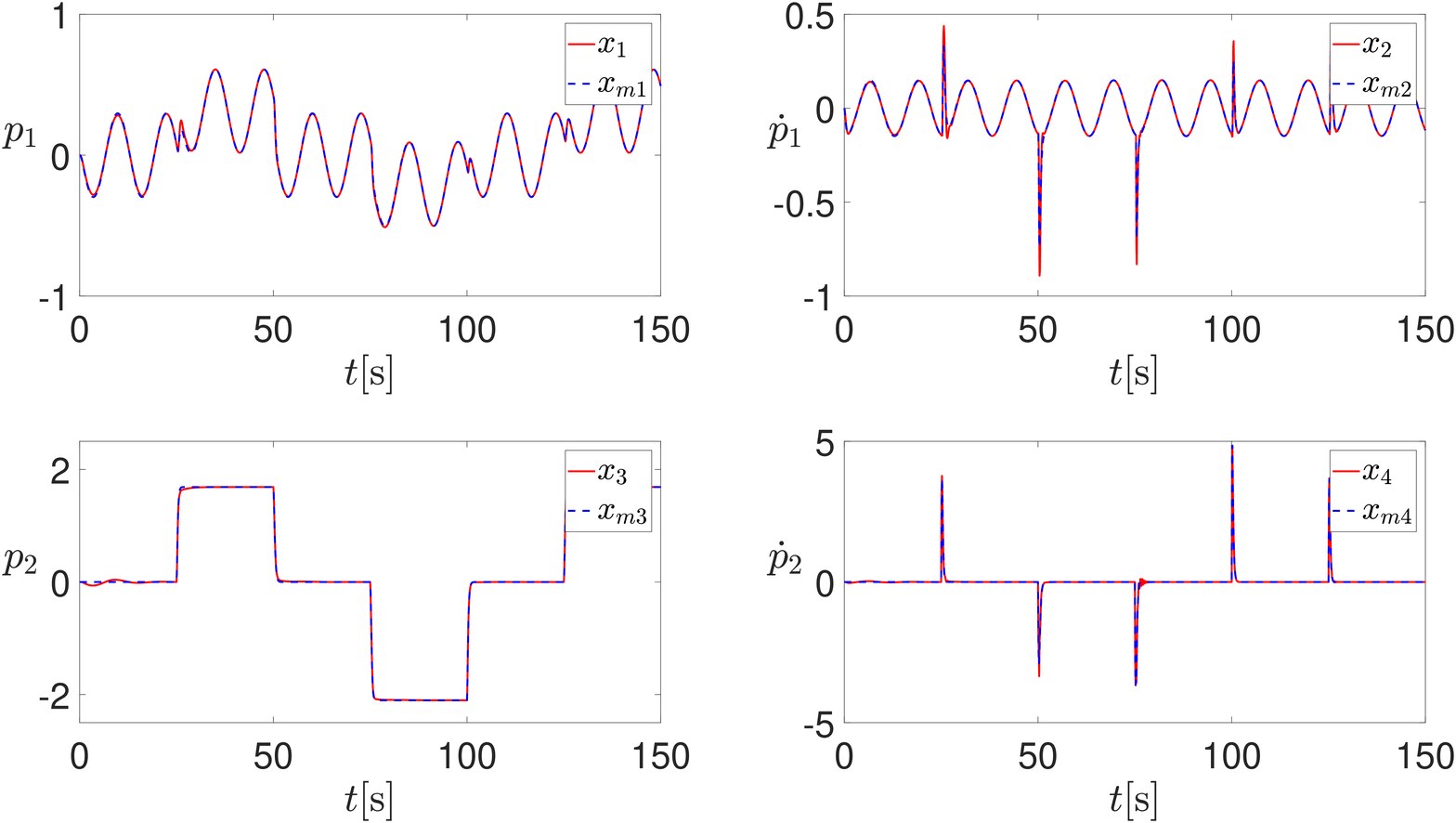}
    \caption{State tracking performance of the proposed MRAC}
    \label{fig: state}
\end{figure}

The Lyapunov function $V$ and the value of the potential function $\phi$ are displayed in Fig. \ref{fig: V}. We observe that the Lyapunov function $V$ is non-increasing, also at the switching instants. This validates the theoretical statement given in Theorem \ref{thm: direct_stability}. As expected, the potential function $\phi$ has jumps at the switching time instants, which is caused by the reset of $\epsilon$ and the value jumps of $\|e\|_P$. We also see that the value of $\phi$ is no larger than $1$, which also reflects that $\|e\|_P < \epsilon$ holds in the given time interval.
\begin{figure}[h]
    \includegraphics[width=1\textwidth]{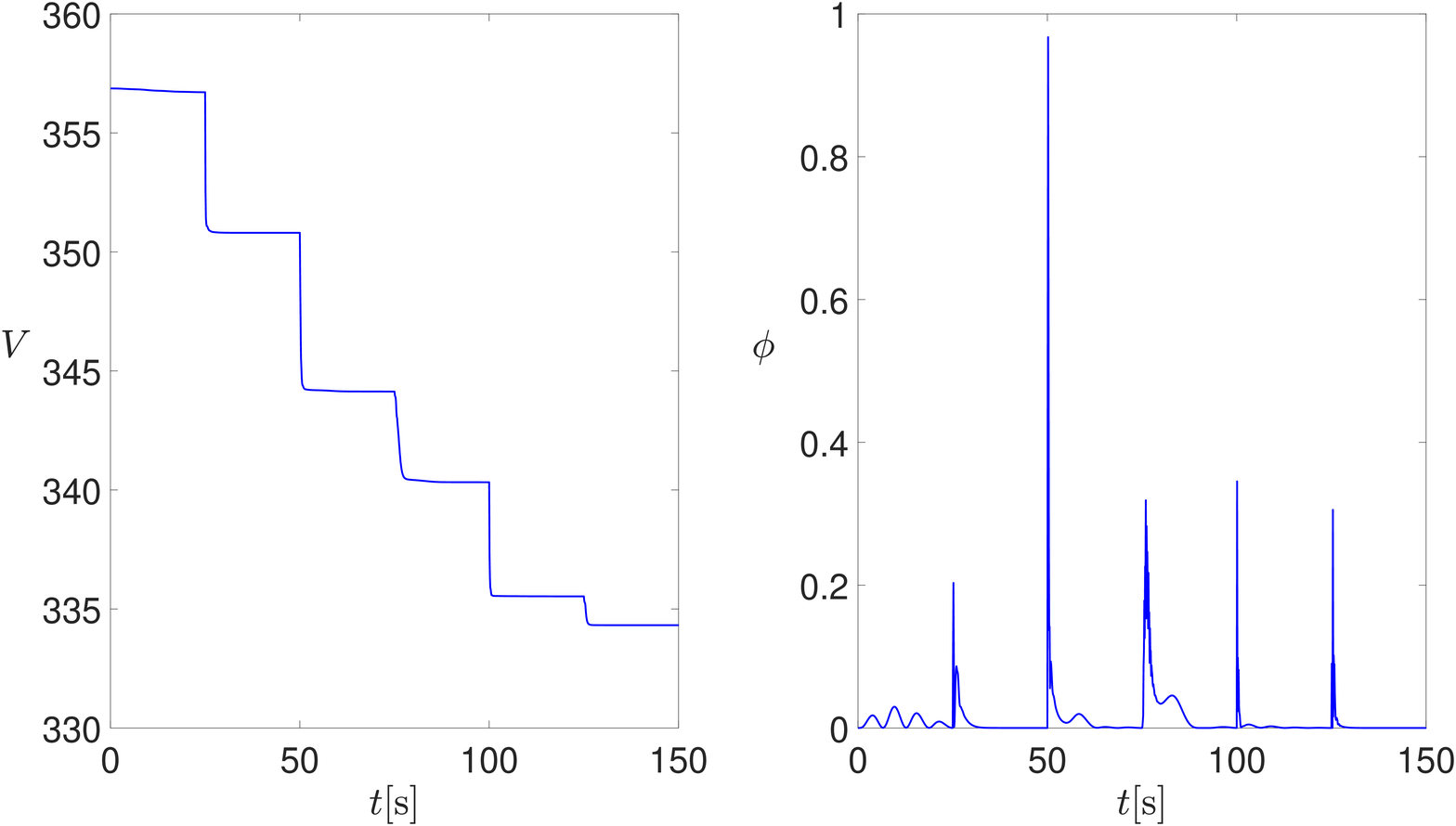}
    \caption{Lyapunov function and the barrier function}
    \label{fig: V}
\end{figure}

\section{Conclusion}
In this paper, we explored MRAC approach for PWA systems with time-varying performance guarantees on the state tracking error. The proposed method is based on barrier functions. To solve the barrier transgression problem caused by the discontinuity of the weighted Euclidean norm of the tracking error, we introduce an auxiliary performance signal, which resides within the performance bound, to construct a barrier function. With state reset at each switching instant, the weighted Euclidean norm of the state tracking error is guaranteed to be confined within the auxiliary performance bound. We construct a Lyapunov function, which is non-increasing even at the switching instants. The dwell time constraints are therefore, dependent only on the user-defined performance bound and the auxiliary performance signal. Future work may include the stability analysis when sliding mode on switching hyperplanes occurs.
\appendix

\section{Proof of Lemma \ref{thm: eps_rho}}
\label{apd: lemma_1}
\begin{proof}
The initial value of $\epsilon$ has $\epsilon(t_0)>\frac{g}{h}$, meaning that $\epsilon$ decreases exponentially towards $\frac{g}{h}$ if no switch occurs. Since $\sqrt{\mu}>1$, $\epsilon$ increases at each switching time instant and $\epsilon(t_k)>\frac{g}{h}$ for $\forall k\in \mathbb{N}^+$. If the switch terminates from some time on, then $\epsilon \to \frac{g}{h}$ for $t \to \infty$, otherwise, $\epsilon>\frac{g}{h}$ for $t \in [t_0,\infty)$. Therefore, we have $\epsilon(t) \geq \frac{g}{h}, \forall t \in [t_0,\infty)$.

Now, we explore the relationship between $\epsilon(t)$ and $\rho(t)$. We have for the time interval $[t_0,t_1)$
\begin{align}
    \begin{split}
        \epsilon(t)=\epsilon(t_0)\mathrm{e}^{-h(t-t_0)}+g\int_{t_0}^t \mathrm{e}^{-h(t-\tau)} \mathrm{d}\tau
        =(\epsilon(t_0)-\frac{g}{h})\mathrm{e}^{-h(t-t_0)}+\frac{g}{h}.
    \end{split}
\end{align}
Since $\epsilon(t_0) \in (\frac{g}{h}, \rho_0)$, $h>l$ and $\rho_\infty>\sqrt{\mu}\frac{g}{h}$, we have $\epsilon(t)<\rho(t)$ for $t\in[t_0,t_1)$. For $t=t_1$ it gives
\begin{align}
\begin{split}
    \epsilon(t_1)=\sqrt{\mu}\epsilon(t_1^-)=\sqrt{\mu}(\epsilon(t_0)-\frac{g}{h})\mathrm{e}^{-h(t_1-t_0)}+\sqrt{\mu}\frac{g}{h}.
\end{split}
\end{align}
Let $\Delta t_1 \triangleq t_1-t_0$, we have
\begin{align}
    \begin{split}
        &\rho(t_1)-\epsilon(t_1)\\
        =&(\rho_0-\rho_\infty)\mathrm{e}^{-l \Delta t_1}-\sqrt{\mu}(\epsilon(t_0)-\frac{g}{h})\mathrm{e}^{-h \Delta t_1}+(\rho_\infty-\sqrt{\mu}\frac{g}{h})\\
        \geq &(\rho_0-\rho_\infty)\mathrm{e}^{-l \Delta t_1}-\sqrt{\mu}(\epsilon(t_0)-\frac{g}{h})\mathrm{e}^{-h \Delta t_1}+(\rho_\infty-\sqrt{\mu}\frac{g}{h})\mathrm{e}^{-l\Delta t_1}\\
        =&(\rho_0-\sqrt{\mu}\frac{g}{h})\mathrm{e}^{-l \Delta t_1}-\sqrt{\mu}(\epsilon(t_0)-\frac{g}{h})\mathrm{e}^{-h \Delta t_1}\\
        \geq &(\rho_0-\sqrt{\mu}\frac{g}{h})\mathrm{e}^{-l \Delta t_1}-\sqrt{\mu}(\rho_0-\frac{g}{h})\mathrm{e}^{-h \Delta t_1}.
    \end{split}
\end{align}
If the following inequality holds, we will immediately have $\rho(t_1)>\epsilon(t_1)$.
\begin{equation}
\label{eqn: ineq_tmp1}
    (\rho_0-\sqrt{\mu}\frac{g}{h})\mathrm{e}^{-l \Delta t_1}>\sqrt{\mu}(\rho_0-\frac{g}{h})\mathrm{e}^{-h \Delta t_1}.
\end{equation}
Since $\rho_0>\rho_\infty>\sqrt{\mu}\frac{g}{h}>\frac{g}{h}$, we have $\rho_0-\sqrt{\mu}\frac{g}{h}>0$ and $\sqrt{\mu}(\rho_0-\frac{g}{h})>0$. Therefore, (\ref{eqn: ineq_tmp1}) is equivalent to
\begin{align}
    \begin{split}
        \frac{\rho_0-\sqrt{\mu}\frac{g}{h}}{\sqrt{\mu}(\rho_0-\frac{g}{h})}>\mathrm{e}^{-(h-l)\Delta t_1}
    \end{split}
\end{align}
Taking the logarithm of both sides we obtain
\begin{equation}
    \Delta t_1 > \frac{1}{h-l}\ln{\frac{\sqrt{\mu} \rho_0-\frac{g}{h}\sqrt{\mu}}{\rho_0-\frac{g}{h}\sqrt{\mu}}}.
\end{equation}
Following the above analysis we can obtain $\epsilon(t) <\rho(t)$ for $t \in [t_{k-1}, t_k)$ and $\epsilon(t_k) <\rho(t_k)$ for $k \in \mathbb{N}^+$ if
\begin{equation}
    \Delta t_k 
    > \frac{1}{h-l}\ln{\frac{\sqrt{\mu} \rho(t_{k-1})-\frac{g}{h}\sqrt{\mu}}{\rho(t_{k-1})-\frac{g}{h}\sqrt{\mu}}}
    =\frac{1}{h-l}\ln({\sqrt{\mu}+\frac{(\mu-\sqrt{\mu})\frac{g}{h}}{\rho(t_{k-1})-\frac{g}{h}\sqrt{\mu}}}).
\end{equation}
If the dwell time $\tau_D$ is no smaller than the maximal required interval length $\max\{\Delta t_k\}$, then $\epsilon(t) < \rho(t)$ holds for $\cup [t_{k-1},t_k), k\in\mathbb{N}^+$. Because $\rho(t_{k-1}) \geq \rho_\infty$ for $k\in\mathbb{N}^+$, we have
\begin{equation}
    \tau_D \geq \max\{\Delta t_k\} = \frac{1}{h-l}\ln{\frac{\sqrt{\mu} \rho_\infty-\frac{g}{h}\sqrt{\mu}}{\rho_\infty-\frac{g}{h}\sqrt{\mu}}}
\end{equation}
So we can conclude that if (\ref{eqn: condition_tau}) holds, then $\epsilon(t) < \rho(t)$ for $t\in[t_0,\infty)$.
\end{proof}

\section{Proof of Lemma \ref{thm: ref_stability}}
\label{apd: lemma_2}
Consider the Lyapunov function $V_m=x_m^T (\sum_{i=1}^s \chi_i P_i) x_m$ for the homogeneous part of (\ref{eqn: ref sys}). The increment of $V_m$ at switching instants satisfies $V_m(t_k) \leq \mu V_m(t_k^-)$. In the interval $t \in [t_{k-1},t_k), k\in \mathbb{N}^+$, we have $\dot{V}_m \leq -\alpha_m V_m$ with 
\begin{equation}
    \alpha_m=\min_{i \in \mathcal{I}}\frac{\lambda_{\text{min}}(Q_i)}{\lambda_{\text{max}}(P_i)}.
\end{equation}
If the switching satisfies $t_k-t_{k-1} > \frac{\ln{\mu}}{\alpha_m}, \forall k\in \mathbb{N}^+$, the homogeneous system $\dot{x}_m=A_m x_m$ is exponentially stable and the stability of the reference system (\ref{eqn: ref sys}) can be concluded for bounded input $r$ (see \cite{morse1996supervisory, hespanha1999stability}). From (\ref{eqn: condition_h}) We have $h-l < h < \frac{1}{2}\alpha_m$, this together with $\mu>1$ leads to
\begin{align}
    \begin{split}
        \tau_D>\frac{2}{\alpha_m}\ln{{\frac{\sqrt{\mu} \rho_\infty-\frac{g}{h}\sqrt{\mu}}{\rho_\infty-\frac{g}{h}\sqrt{\mu}}}}>\frac{2}{\alpha_m}\ln{{\frac{\sqrt{\mu} (\rho_\infty-\frac{g}{h})}{\rho_\infty-\frac{g}{h}}}}=\frac{\ln{\mu}}{\alpha_m}.
    \end{split}
\end{align}
So this tells that the reference system is stable and $x_m \in \mathcal{L}_\infty$ if the dwell time constraint $\tau_D$ in (\ref{eqn: condition_tau}) is satisfied. 
\bibliographystyle{elsarticle-harv} 
\bibliography{main.bib}

\begin{thebibliography}{37}
\expandafter\ifx\csname natexlab\endcsname\relax\def\natexlab#1{#1}\fi
\providecommand{\url}[1]{\texttt{#1}}
\providecommand{\href}[2]{#2}
\providecommand{\path}[1]{#1}
\providecommand{\DOIprefix}{doi:}
\providecommand{\ArXivprefix}{arXiv:}
\providecommand{\URLprefix}{URL: }
\providecommand{\Pubmedprefix}{pmid:}
\providecommand{\doi}[1]{\href{http://dx.doi.org/#1}{\path{#1}}}
\providecommand{\Pubmed}[1]{\href{pmid:#1}{\path{#1}}}
\providecommand{\bibinfo}[2]{#2}
\ifx\xfnm\relax \def\xfnm[#1]{\unskip,\space#1}\fi
\bibitem[{Arabi et~al.(2018)Arabi, Gruenwald, Yucelen and
  Nguyen}]{arabi2018set}
\bibinfo{author}{Arabi, E.}, \bibinfo{author}{Gruenwald, B.C.},
  \bibinfo{author}{Yucelen, T.}, \bibinfo{author}{Nguyen, N.T.},
  \bibinfo{year}{2018}.
\newblock \bibinfo{title}{A set-theoretic model reference adaptive control
  architecture for disturbance rejection and uncertainty suppression with
  strict performance guarantees}.
\newblock \bibinfo{journal}{International Journal of Control}
  \bibinfo{volume}{91}, \bibinfo{pages}{1195--1208}.
\bibitem[{Arabi and Yucelen(2019)}]{arabi2019set}
\bibinfo{author}{Arabi, E.}, \bibinfo{author}{Yucelen, T.},
  \bibinfo{year}{2019}.
\newblock \bibinfo{title}{Set-theoretic model reference adaptive control with
  time-varying performance bounds}.
\newblock \bibinfo{journal}{International Journal of Control}
  \bibinfo{volume}{92}, \bibinfo{pages}{2509--2520}.
\bibitem[{Arabi et~al.(2019)Arabi, Yucelen, Gruenwald, Fravolini, Balakrishnan
  and Nguyen}]{arabi2019neuroadaptive}
\bibinfo{author}{Arabi, E.}, \bibinfo{author}{Yucelen, T.},
  \bibinfo{author}{Gruenwald, B.C.}, \bibinfo{author}{Fravolini, M.},
  \bibinfo{author}{Balakrishnan, S.}, \bibinfo{author}{Nguyen, N.T.},
  \bibinfo{year}{2019}.
\newblock \bibinfo{title}{A neuroadaptive architecture for model reference
  control of uncertain dynamical systems with performance guarantees}.
\newblock \bibinfo{journal}{Systems \& Control Letters} \bibinfo{volume}{125},
  \bibinfo{pages}{37--44}.
\bibitem[{Bechlioulis and Rovithakis(2008)}]{bechlioulis2008robust}
\bibinfo{author}{Bechlioulis, C.P.}, \bibinfo{author}{Rovithakis, G.A.},
  \bibinfo{year}{2008}.
\newblock \bibinfo{title}{Robust adaptive control of feedback linearizable mimo
  nonlinear systems with prescribed performance}.
\newblock \bibinfo{journal}{IEEE Transactions on Automatic Control}
  \bibinfo{volume}{53}, \bibinfo{pages}{2090--2099}.
\bibitem[{Bechlioulis and Rovithakis(2010)}]{bechlioulis2010prescribed}
\bibinfo{author}{Bechlioulis, C.P.}, \bibinfo{author}{Rovithakis, G.A.},
  \bibinfo{year}{2010}.
\newblock \bibinfo{title}{Prescribed performance adaptive control for
  multi-input multi-output affine in the control nonlinear systems}.
\newblock \bibinfo{journal}{IEEE Transactions on Automatic Control}
  \bibinfo{volume}{55}, \bibinfo{pages}{1220--1226}.
\bibitem[{Bemporad et~al.(2000)Bemporad, Ferrari-Trecate and
  Morari}]{bemporad2000observability}
\bibinfo{author}{Bemporad, A.}, \bibinfo{author}{Ferrari-Trecate, G.},
  \bibinfo{author}{Morari, M.}, \bibinfo{year}{2000}.
\newblock \bibinfo{title}{Observability and controllability of piecewise affine
  and hybrid systems}.
\newblock \bibinfo{journal}{IEEE transactions on automatic control}
  \bibinfo{volume}{45}, \bibinfo{pages}{1864--1876}.
\bibitem[{di~Bernardo et~al.(2016)di~Bernardo, Montanaro, Ortega and
  Santini}]{di2016extended}
\bibinfo{author}{di~Bernardo, M.}, \bibinfo{author}{Montanaro, U.},
  \bibinfo{author}{Ortega, R.}, \bibinfo{author}{Santini, S.},
  \bibinfo{year}{2016}.
\newblock \bibinfo{title}{Extended hybrid model reference adaptive control of
  piecewise affine systems}.
\newblock \bibinfo{journal}{Nonlinear Analysis: Hybrid Systems}
  \bibinfo{volume}{21}, \bibinfo{pages}{11--21}.
\bibitem[{di~Bernardo et~al.(2013)di~Bernardo, Montanaro and
  Santini}]{di2013hybrid}
\bibinfo{author}{di~Bernardo, M.}, \bibinfo{author}{Montanaro, U.},
  \bibinfo{author}{Santini, S.}, \bibinfo{year}{2013}.
\newblock \bibinfo{title}{Hybrid model reference adaptive control of piecewise
  affine systems}.
\newblock \bibinfo{journal}{IEEE Transactions on Automatic Control}
  \bibinfo{volume}{58}, \bibinfo{pages}{304--316}.
\bibitem[{Bernardo et~al.(2013)Bernardo, Montanaro, Olm and
  Santini}]{bernardo2013model}
\bibinfo{author}{Bernardo, M.d.}, \bibinfo{author}{Montanaro, U.},
  \bibinfo{author}{Olm, J.M.}, \bibinfo{author}{Santini, S.},
  \bibinfo{year}{2013}.
\newblock \bibinfo{title}{Model reference adaptive control of discrete-time
  piecewise linear systems}.
\newblock \bibinfo{journal}{International Journal of Robust and Nonlinear
  Control} \bibinfo{volume}{23}, \bibinfo{pages}{709--730}.
\bibitem[{Collins and Van~Schuppen(2004)}]{collins2004observability}
\bibinfo{author}{Collins, P.}, \bibinfo{author}{Van~Schuppen, J.H.},
  \bibinfo{year}{2004}.
\newblock \bibinfo{title}{Observability of piecewise-affine hybrid systems},
  in: \bibinfo{booktitle}{International Workshop on Hybrid Systems: Computation
  and Control}, \bibinfo{organization}{Springer}. pp.
  \bibinfo{pages}{265--279}.
\bibitem[{Habets et~al.(2006)Habets, Collins and van
  Schuppen}]{habets2006reachability}
\bibinfo{author}{Habets, L.}, \bibinfo{author}{Collins, P.J.},
  \bibinfo{author}{van Schuppen, J.H.}, \bibinfo{year}{2006}.
\newblock \bibinfo{title}{Reachability and control synthesis for
  piecewise-affine hybrid systems on simplices}.
\newblock \bibinfo{journal}{IEEE Transactions on Automatic Control}
  \bibinfo{volume}{51}, \bibinfo{pages}{938--948}.
\bibitem[{Hackl et~al.(2013)Hackl, Hopfe, Ilchmann, Mueller and
  Trenn}]{hackl2013funnel}
\bibinfo{author}{Hackl, C.M.}, \bibinfo{author}{Hopfe, N.},
  \bibinfo{author}{Ilchmann, A.}, \bibinfo{author}{Mueller, M.},
  \bibinfo{author}{Trenn, S.}, \bibinfo{year}{2013}.
\newblock \bibinfo{title}{Funnel control for systems with relative degree two}.
\newblock \bibinfo{journal}{SIAM Journal on Control and Optimization}
  \bibinfo{volume}{51}, \bibinfo{pages}{965--995}.
\bibitem[{Hespanha and Morse(1999)}]{hespanha1999stability}
\bibinfo{author}{Hespanha, J.P.}, \bibinfo{author}{Morse, A.S.},
  \bibinfo{year}{1999}.
\newblock \bibinfo{title}{Stability of switched systems with average
  dwell-time}, in: \bibinfo{booktitle}{Proceedings of the 38th IEEE conference
  on decision and control (Cat. No. 99CH36304)}, \bibinfo{organization}{IEEE}.
  pp. \bibinfo{pages}{2655--2660}.
\bibitem[{Ilchmann and Schuster(2009)}]{ilchmann2009pi}
\bibinfo{author}{Ilchmann, A.}, \bibinfo{author}{Schuster, H.},
  \bibinfo{year}{2009}.
\newblock \bibinfo{title}{Pi-funnel control for two mass systems}.
\newblock \bibinfo{journal}{IEEE Transactions on Automatic Control}
  \bibinfo{volume}{54}, \bibinfo{pages}{918--923}.
\bibitem[{Kersting and Buss(2017)}]{kersting2017direct}
\bibinfo{author}{Kersting, S.}, \bibinfo{author}{Buss, M.},
  \bibinfo{year}{2017}.
\newblock \bibinfo{title}{Direct and indirect model reference adaptive control
  for multivariable piecewise affine systems}.
\newblock \bibinfo{journal}{IEEE Transactions on Automatic Control}
  \bibinfo{volume}{62}, \bibinfo{pages}{5634--5649}.
\bibitem[{Lavretsky(2011)}]{lavretsky2011adaptive}
\bibinfo{author}{Lavretsky, E.}, \bibinfo{year}{2011}.
\newblock \bibinfo{title}{Adaptive output feedback design using asymptotic
  properties of lqg/ltr controllers}.
\newblock \bibinfo{journal}{IEEE Transactions on Automatic Control}
  \bibinfo{volume}{57}, \bibinfo{pages}{1587--1591}.
\bibitem[{Liu et~al.(2014)Liu, Li and Tong}]{liu2014adaptive}
\bibinfo{author}{Liu, Y.J.}, \bibinfo{author}{Li, D.J.}, \bibinfo{author}{Tong,
  S.}, \bibinfo{year}{2014}.
\newblock \bibinfo{title}{Adaptive output feedback control for a class of
  nonlinear systems with full-state constraints}.
\newblock \bibinfo{journal}{International Journal of Control}
  \bibinfo{volume}{87}, \bibinfo{pages}{281--290}.
\bibitem[{Liu and Tong(2016)}]{liu2016barrier}
\bibinfo{author}{Liu, Y.J.}, \bibinfo{author}{Tong, S.}, \bibinfo{year}{2016}.
\newblock \bibinfo{title}{Barrier lyapunov functions-based adaptive control for
  a class of nonlinear pure-feedback systems with full state constraints}.
\newblock \bibinfo{journal}{Automatica} \bibinfo{volume}{64},
  \bibinfo{pages}{70--75}.
\bibitem[{Morse(1996)}]{morse1996supervisory}
\bibinfo{author}{Morse, A.S.}, \bibinfo{year}{1996}.
\newblock \bibinfo{title}{Supervisory control of families of linear set-point
  controllers-part i. exact matching}.
\newblock \bibinfo{journal}{IEEE transactions on Automatic Control}
  \bibinfo{volume}{41}, \bibinfo{pages}{1413--1431}.
\bibitem[{Niu et~al.(2015)Niu, Zhao, Fan and Cheng}]{niu2015new}
\bibinfo{author}{Niu, B.}, \bibinfo{author}{Zhao, X.}, \bibinfo{author}{Fan,
  X.}, \bibinfo{author}{Cheng, Y.}, \bibinfo{year}{2015}.
\newblock \bibinfo{title}{A new control method for state-constrained nonlinear
  switched systems with application to chemical process}.
\newblock \bibinfo{journal}{International Journal of Control}
  \bibinfo{volume}{88}, \bibinfo{pages}{1693--1701}.
\bibitem[{Pavlov et~al.(2007)Pavlov, Pogromsky, Van De~Wouw and
  Nijmeijer}]{pavlov2007convergence}
\bibinfo{author}{Pavlov, A.}, \bibinfo{author}{Pogromsky, A.},
  \bibinfo{author}{Van De~Wouw, N.}, \bibinfo{author}{Nijmeijer, H.},
  \bibinfo{year}{2007}.
\newblock \bibinfo{title}{On convergence properties of piecewise affine
  systems}.
\newblock \bibinfo{journal}{International Journal of Control}
  \bibinfo{volume}{80}, \bibinfo{pages}{1233--1247}.
\bibitem[{Rodrigues and How(2003)}]{rodrigues2003observer}
\bibinfo{author}{Rodrigues, L.}, \bibinfo{author}{How, J.P.},
  \bibinfo{year}{2003}.
\newblock \bibinfo{title}{Observer-based control of piecewise-affine systems}.
\newblock \bibinfo{journal}{International Journal of Control}
  \bibinfo{volume}{76}, \bibinfo{pages}{459--477}.
\bibitem[{Sang and Tao(2011a)}]{sang2011adaptive2}
\bibinfo{author}{Sang, Q.}, \bibinfo{author}{Tao, G.}, \bibinfo{year}{2011}a.
\newblock \bibinfo{title}{Adaptive control of piecewise linear systems with
  applications to nasa gtm}, in: \bibinfo{booktitle}{Proceedings of the 2011
  American Control Conference}, \bibinfo{organization}{IEEE}. pp.
  \bibinfo{pages}{1157--1162}.
\bibitem[{Sang and Tao(2011b)}]{sang2011adaptive}
\bibinfo{author}{Sang, Q.}, \bibinfo{author}{Tao, G.}, \bibinfo{year}{2011}b.
\newblock \bibinfo{title}{Adaptive control of piecewise linear systems with
  applications to nasa gtm}, in: \bibinfo{booktitle}{Proceedings of the 2011
  American Control Conference}, \bibinfo{organization}{IEEE}. pp.
  \bibinfo{pages}{1157--1162}.
\bibitem[{Sang and Tao(2012a)}]{sang2012adaptive2}
\bibinfo{author}{Sang, Q.}, \bibinfo{author}{Tao, G.}, \bibinfo{year}{2012}a.
\newblock \bibinfo{title}{Adaptive control of piecewise linear systems: the
  state tracking case}.
\newblock \bibinfo{journal}{IEEE Transactions on Automatic Control}
  \bibinfo{volume}{57}, \bibinfo{pages}{522--528}.
\bibitem[{Sang and Tao(2012b)}]{sang2012adaptive1}
\bibinfo{author}{Sang, Q.}, \bibinfo{author}{Tao, G.}, \bibinfo{year}{2012}b.
\newblock \bibinfo{title}{Adaptive control of piecewise linear systems with
  output feedback for output tracking}, in: \bibinfo{booktitle}{Decision and
  Control (CDC), 2012 IEEE 51st Annual Conference on},
  \bibinfo{organization}{IEEE}. pp. \bibinfo{pages}{5422--5427}.
\bibitem[{Tao(2014)}]{tao2014multivariable}
\bibinfo{author}{Tao, G.}, \bibinfo{year}{2014}.
\newblock \bibinfo{title}{Multivariable adaptive control: A survey}.
\newblock \bibinfo{journal}{Automatica} \bibinfo{volume}{50},
  \bibinfo{pages}{2737--2764}.
\bibitem[{Tao et~al.(2020)Tao, Roy and Baldi}]{tao2020issue}
\bibinfo{author}{Tao, T.}, \bibinfo{author}{Roy, S.}, \bibinfo{author}{Baldi,
  S.}, \bibinfo{year}{2020}.
\newblock \bibinfo{title}{The issue of transients in leakage-based model
  reference adaptive control of switched linear systems}.
\newblock \bibinfo{journal}{Nonlinear Analysis: Hybrid Systems}
  \bibinfo{volume}{36}, \bibinfo{pages}{100885}.
\bibitem[{Tee et~al.(2009)Tee, Ge and Tay}]{tee2009barrier}
\bibinfo{author}{Tee, K.P.}, \bibinfo{author}{Ge, S.S.}, \bibinfo{author}{Tay,
  E.H.}, \bibinfo{year}{2009}.
\newblock \bibinfo{title}{Barrier lyapunov functions for the control of
  output-constrained nonlinear systems}.
\newblock \bibinfo{journal}{Automatica} \bibinfo{volume}{45},
  \bibinfo{pages}{918--927}.
\bibitem[{Wang et~al.(2012)Wang, Hou and Dong}]{wang2012model}
\bibinfo{author}{Wang, Q.}, \bibinfo{author}{Hou, Y.}, \bibinfo{author}{Dong,
  C.}, \bibinfo{year}{2012}.
\newblock \bibinfo{title}{Model reference robust adaptive control for a class
  of uncertain switched linear systems}.
\newblock \bibinfo{journal}{International Journal of Robust and Nonlinear
  Control} \bibinfo{volume}{22}, \bibinfo{pages}{1019--1035}.
\bibitem[{Wu and Zhao(2015)}]{wu2015h}
\bibinfo{author}{Wu, C.}, \bibinfo{author}{Zhao, J.}, \bibinfo{year}{2015}.
\newblock \bibinfo{title}{$h_\infty$ adaptive tracking control for switched
  systems based on an average dwell-time method}.
\newblock \bibinfo{journal}{International Journal of Systems Science}
  \bibinfo{volume}{46}, \bibinfo{pages}{2547--2559}.
\bibitem[{Wu et~al.(2015)Wu, Zhao and Sun}]{wu2015adaptive}
\bibinfo{author}{Wu, C.}, \bibinfo{author}{Zhao, J.}, \bibinfo{author}{Sun,
  X.M.}, \bibinfo{year}{2015}.
\newblock \bibinfo{title}{Adaptive tracking control for uncertain switched
  systems under asynchronous switching}.
\newblock \bibinfo{journal}{International Journal of Robust and Nonlinear
  Control} \bibinfo{volume}{25}, \bibinfo{pages}{3457--3477}.
\bibitem[{Xiao and Dong(2019)}]{xiao2019robust}
\bibinfo{author}{Xiao, S.}, \bibinfo{author}{Dong, J.}, \bibinfo{year}{2019}.
\newblock \bibinfo{title}{Robust adaptive fault-tolerant tracking control for
  uncertain linear systems with time-varying performance bounds}.
\newblock \bibinfo{journal}{International Journal of Robust and Nonlinear
  Control} \bibinfo{volume}{29}, \bibinfo{pages}{849--866}.
\bibitem[{Xie and Zhao(2018)}]{xie2018h}
\bibinfo{author}{Xie, J.}, \bibinfo{author}{Zhao, J.}, \bibinfo{year}{2018}.
\newblock \bibinfo{title}{$h_\infty$ model reference adaptive control for
  switched systems based on the switched closed-loop reference model}.
\newblock \bibinfo{journal}{Nonlinear Analysis: Hybrid Systems}
  \bibinfo{volume}{27}, \bibinfo{pages}{92--106}.
\bibitem[{Yuan et~al.(2018a)Yuan, De~Schutter and Baldi}]{yuan2018robust}
\bibinfo{author}{Yuan, S.}, \bibinfo{author}{De~Schutter, B.},
  \bibinfo{author}{Baldi, S.}, \bibinfo{year}{2018}a.
\newblock \bibinfo{title}{Robust adaptive tracking control of uncertain slowly
  switched linear systems}.
\newblock \bibinfo{journal}{Nonlinear Analysis: Hybrid Systems}
  \bibinfo{volume}{27}, \bibinfo{pages}{1--12}.
\bibitem[{Yuan et~al.(2018b)Yuan, Zhang, De~Schutter and Baldi}]{yuan2018novel}
\bibinfo{author}{Yuan, S.}, \bibinfo{author}{Zhang, L.},
  \bibinfo{author}{De~Schutter, B.}, \bibinfo{author}{Baldi, S.},
  \bibinfo{year}{2018}b.
\newblock \bibinfo{title}{A novel lyapunov function for a non-weighted l2 gain
  of asynchronously switched linear systems}.
\newblock \bibinfo{journal}{Automatica} \bibinfo{volume}{87},
  \bibinfo{pages}{310--317}.
\bibitem[{Zhao and Song(2018)}]{zhao2018removing}
\bibinfo{author}{Zhao, K.}, \bibinfo{author}{Song, Y.}, \bibinfo{year}{2018}.
\newblock \bibinfo{title}{Removing the feasibility conditions imposed on
  tracking control designs for state-constrained strict-feedback systems}.
\newblock \bibinfo{journal}{IEEE Transactions on Automatic Control}
  \bibinfo{volume}{64}, \bibinfo{pages}{1265--1272}.

\end{thebibliography}





\end{document}